\documentclass[11pt]{article}
\usepackage{amssymb,amsfonts,mathtools}
\usepackage{enumerate}
\usepackage{amsmath,amsthm,amssymb,graphicx,marvosym,wasysym,appendix}
\usepackage{algorithm, color}
\usepackage{algorithmic}
\usepackage{natbib}
\usepackage{booktabs} 
\usepackage{setspace,bm}
\usepackage{hyperref}
\hypersetup{%
  colorlinks=false,
  pdfborder = {0 0 0.5 [3 0]}
}

\usepackage{graphicx}
\usepackage{epstopdf}
\newcommand{\indic}[1]{1\hspace{-2.1mm}{1}_{\{#1\}}} 
\newcommand{\PP}{{\mathord{I\kern -.33em P}}}
\newcommand{\EE}{{\mathord{I\kern -.33em E}}}
\newcommand{\RR}{{\mathord{I\kern -.33em R}}}
\def\R{\mathbb{R}} 
\def\P{\mathbb{Q}} 
\def\E{\mathbb{E}} 
\def\N{{\mathbb N}} 
\def\Q{{\mathbb Q}} 
\def\1{{\mathbf 1}} 
\def\F{{\mathcal F}} 
\def\D{{\mathcal D}}

\def\tW{{\tilde{W}}}

\def\L{{\mathcal L}}

\def\M{{\mathcal M}}

\def\tlambda{{\tilde{\lambda}}}

\def\th{{\tilde{h}}}

\def\hP{{\widehat{P}}}

\def\Spread{{\mathcal{S}}}

\newtheorem{theorem}{Theorem}[section]

\newtheorem{lemma}[theorem]{Lemma}
\newtheorem{example}[theorem]{Example}

\newtheorem{proposition}[theorem]{Proposition}

\theoremstyle{remark}
\newtheorem{remark}[theorem]{Remark}

\numberwithin{equation}{section}
\numberwithin{theorem}{section}
\begin{document}

\title{Pricing Derivatives with
Counterparty Risk and Collateralization:\\ A Fixed Point Approach}
\author{Jinbeom Kim\thanks{Industrial Engineering and Operations Research Department, Columbia University, New York NY 10027. E-mail:  \mbox{{jk3071@columbia.edu}}. } \qquad Tim Leung\thanks{Corresponding author. Industrial Engineering and Operations Research Department,  Columbia University, New York NY 10027. E-mail:  \mbox{{leung@ieor.columbia.edu}}. }}
\date{January 25, 2015} \maketitle
\begin{abstract}
\begin{small}
This paper  studies a  valuation framework  for   financial contracts subject to reference and counterparty default risks with collateralization requirement. We propose a fixed point  approach  to analyze  the mark-to-market contract value with  counterparty risk provision,   and show that it   is a unique bounded and continuous fixed point via  contraction mapping. This leads us to develop an accurate   iterative  numerical scheme for valuation.   Specifically, we solve   a sequence of linear  inhomogeneous PDEs, whose solutions  converge   to the fixed point price function.   We apply our methodology to compute the bid and ask prices for  both defaultable equity and fixed-income derivatives, and illustrate the non-trivial effects  of counterparty risk, collateralization ratio and liquidation convention  on the bid-ask spreads.\end{small}
\end{abstract}
\vspace{10pt}
\begin{small}
\end{small}
\begin{small}
\tableofcontents
\end{small}

\onehalfspacing
\section{Introduction} \label{sec:intro}
Counterparty risk has played an important role during the 2008 financial crisis. According to the Bank for International Settlements (BIS)\footnote{See BIS press release at http://www.bis.org/press/p110601.pdf},   two-thirds of counterparty risk losses during the crisis were from counterparty risk adjustments in MtM valuation whereas the rest were due to actual defaults.  In order to account for the counterparty risk, recent regulatory changes, such as Basel III, incorporate the counterparty risk adjustments in the calculation of capital requirement. On the other hand, the use of collateral in the derivative market has increased dramatically. According to the survey conducted by the International {S}waps and {D}erivatives {A}ssociation ({I}{S}{D}{A}) in 2013\footnote{Survey available at http://www2.isda.org/functional-areas/research/surveys/margin-surveys/}, the percentage of all trades subject to collateral agreements in the over-the-counter (OTC) market increases from $30\%$ in 2003 to $73.7\%$ in 2013.  OTC market participants continue to adapt collateralization and  counterparty risk adjustments in their MtM valuation methodologies for various contracts, including forwards, total return swaps, interest rate swaps  and credit default swaps.

When an OTC market participant trades a financial claim with a counterparty, the participant is exposed not only to the price change and default risk of the underlying asset but also to the default risk of the counterparty. To reflect the counterparty default risk in MtM valuation,  three adjustments are calculated in addition to the counterparty-risk free value of the claim. While credit valuation adjustment (CVA) accounts for the possibility of the counterparty's default, debt valuation adjustment (DVA) is calculated to adjust for the participant's own default risk. In addition, collateral interest payments and the cost of borrowing  generate funding valuation adjustment (FVA). In the industry, the valuation  adjustment incorporating CVA, DVA, FVA and  collateralization,  is called total valuation adjustment (XVA)\footnote{The terminology can be found in \cite{XVA_Riskmag} and \cite{fouque2013handbook}, among others. }.

In this paper, we study a valuation framework  for financial contracts accounting for XVA. We consider two current market conventions for price computation. The main difference in the two conventions rises in the assumption of the liquidation value -- either counterparty risk-free value or MtM value with counterparty risk provision -- upon default. \cite{brigo2012counterparty} and  \cite{brigo2011credit} show that the values under the two conventions have significant differences and large impacts on net debtors and creditors.

With counterparty risk provision, the MtM value is defined implicitly via a risk-neutral expectation. This gives rise to major challenges in analyzing and computing the contract value. We propose a novel  fixed point approach to analyze the MtM value. Our methodology  involves solving a sequence of inhomogeneous linear PDEs,  whose classical solutions are shown via contraction mapping arguments to  converge to the unique fixed point price function.  This approach also motivates us to develop an iterative numerical scheme to compute the  values of  a variety of financial claims under different market conventions. 

In related studies,  \cite{fujii2013derivative} incorporate BCVA and under/over collateralization, and  calculate  the MtM value  by simulation.  \cite{henry2012counterparty} approximates the MtM value by numerically solving  a related nonlinear PDE  through simulation of a marked branching diffusion, and provides  conditions to   avoid a ``blow-up" of the simulated solution. \cite{burgard2011partial} also consider a similar nonlinear PDE in their study of hedging strategies inclusive of funding costs, and taking into account closeout payments exchanged at the time when either party of the contract defaults.  In contrast, our fixed point methodology works directly with  the   price  definition in terms of a recursive expectation, rather than heuristically  stating and solving  a nonlinear PDE.  Our  contraction mapping result allows us to solve a series of \textit{linear} PDE problems with  bounded classical solutions, and obtain  a unique bounded continuous   MtM value.

Our model also provides insight on the bid-ask prices of various financial contracts. The XVA is asymmetric for the buyer and the seller. As such, the incorporation of  adjustment to   unilateral   or bilateral counterparty risk leads to  a non-zero bid-ask spread. In other words, counterparty risk reveals itself as a market friction, resulting in a transaction cost for OTC trades.  In addition, we  examine the  impact of various parameters such as  default rate, recovery rate,  collateralization ratio and effective collateral interest rate. We find that a higher  counterparty default rate and funding cost  reduce the  MtM value,  whereas the market participant's own default rate and collateralization  ratio have  positive price effects. For claims with a positive payoff, such as calls and puts,  we establish a number of price dominance relationships.  In particular, when collateral rates are low, the bid-ask prices are dominated by the counterparty risk-free value. Moreover, the bid-ask prices decrease when we use the MtM value rather than counterparty risk-free value for the liquidation value upon default.

The recent regulatory   changes and  post-crisis perception of counterparty risk have motivated  research on the analysis of XVA.  \cite{brigo2012arbitrage} consider an arbitrage-free valuation framework with  bilateral counterparty risk and  collateral with possible rehypothecation. \cite{fouque2013handbook} studies the arbitrage-free valuation  of counterparty risk,  and analyzes   the impact  of default correlation, collateral migration frequency and collateral re-hypothecation on the  collateralized CVA. \cite{thompson2010counterparty} analyzes the effect of counterparty risk on insurance contracts  and examines the moral hazard of the insurers. \cite{brigo2009counterparty} focus on the valuation of  CDS  with counterparty risk that is correlated with the reference default.   \cite{hull2012cva}  investigate  the  wrong way risk -- the additional risk generated by the correlation between  the portfolio return and  counterparty default risk.

Let us give an outline of the rest of the paper. In Sect. \ref{sec:equity_counterparty}, we formulate the MtM valuation of a generic financial claim  with default risk and counterparty default risks under collateralization. In Sect. \ref{sec:fixed_point}, we provide a fixed point theorem and a recursive algorithm for valuation. In Sect. \ref{Equity_Claim}, we compute  the MtM values of various defaultable equity claims and derive their  bid-ask prices.  In Sect. \ref{sec:BCVA_Credit}, we apply our model to  price a number of  defaultable fixed-income claims. Sect. \ref{Conclusions} concludes the paper, and the Appendix  contains a number of longer proofs.
  
\section{Model Formulation}  \label{sec:equity_counterparty}
In the background, we fix a  probability space $(\Omega, \F, \Q)$, where $\Q$ is the risk-neutral pricing measure.  In our model, there are three defaultable parties: a  reference entity, a market participant, and  a counterparty dealer. We denote them respectively  as parties 0, 1, and 2. The default  time $\tau_i$ of party $i \in \{0,1,2\}$ is modeled by the first jump time of an  exogenous doubly stochastic Poisson process. Precisely, we define
\begin{align}
 \tau_i = \inf\big\{t \ge 0 : \int_0^t \lambda^{(i)}_u \, du > E_i\big\}\,, \label{def:default_time_dist}
\end{align}
where $\{E_i\}_{i=0,1,2}$ are unit exponential random variables that are independent of the intensity  processes  $(\lambda^{(i)}_t)_{t \ge 0}$,\,$i \in \{0,1,2\}$. Throughout, each intensity process  is assumed to be of   Markovian form $\lambda^{(i)}_t \equiv \lambda^{(i)}(t,S_t,X_t)$ for some bounded positive  function $\lambda^{(i)}(t,s,x)$, and is driven by the \emph{pre-default} stock price $S$ and the stochastic factor $X$ satisfying the SDEs
\begin{align}
 dS_t  &= (r(t, X_t) +  \lambda^{(0)}(t,S_t,X_t)) \, S_t \, dt + \sigma(t, S_t) \, S_t \,  dW_t\,,  \label{model:stock}\\
  dX_t &= b(t,X_t) \, dt + \eta(t,X_t) \, d\tW_t \,. \label{model:intensity}
\end{align}
Here, $(W_t)_{t \ge 0}$ and $(\tW_t)_{t \ge 0}$ are standard Brownian motions under $\P$ with an instantaneous  correlation parameter $\rho \in(-1, 1)$. The risk-free interest rate is denoted by $r_t 
\equiv r(t,X_t)$ for some bounded positive  function. At the  default time $\tau_0$, the stock price will  jump to value zero and remain worthless afterwards. This ``jump-to-default model" for $S$ is a variation of those by \cite{merton76},  \cite{carr2006jump}, and \cite{RafaLin11}.

\subsection{Mark-to-Market Value with   Counterparty Risk Provision}\label{subsec:price_formulation} A  defaultable claim is described by the triplet $(g,h,l)$, where $g(S_T, X_T)$ is the  payoff at maturity $T$, $(h(S_t, X_t))_{0 \le t \le T}$ is the dividend process, and $l(\tau_0,X_{\tau_0})$ is the payoff at the default time $\tau_0$ of the reference entity. We assume continuous collateralization which is a reasonable proxy for the current market where daily or intraday margin calls are common \citep[see][]{fujii2013derivative}. For party $i \in \{1,2\}$, we denote by $\delta_{i}$ the collateral coverage ratio of the claim's MtM value. We use the range $0 \le \delta_{i} \le 120\%$ since dealers usually require over-collateralization up to $120\%$ for credit or equity linked notes \citep[see][Table 1]{ramaswamy2011market}.

We first consider pricing of a defaultable claim  without bilateral counterparty risk. We call this value \textit{counterparty-risk free (CRF) value}.  Precisely, the ex-dividend \textit{pre-default} time $t$ CRF value of the defaultable claim with $(g,h,l)$ is given by
\begin{align}
\Pi(t,s,x) &:= \E_{t,s,x} \bigg[ \, e^{- \int_t^T (r_u + \lambda^{(0)}_u) \, du} \, g(S_T,X_T) + \int_t^T e^{- \int_t^u (r_v + \lambda^{(0)}_v) \, dv} \, \left(h(S_u,X_u) + \lambda^{(0)}_u \, l(u,X_u)\right) \, du   \bigg]. \label{Expectation_Pi}
\end{align}
The shorthand notation $\E_{t,s,x}[\,\cdot\,] := \E[\, \cdot \, \vert S_t = s ,\, X_t = x \,]$ denotes the conditional expectation under $\P$ given $S_t = s$,\, $X_t = x$.

Incorporating counterparty risk, we let $\tau = \min\{\tau_0,\tau_1, \tau_2 \}$, which is the first default  time  among the three parties with the intensity function
   $\lambda(t,s,x)   =   \sum_{k= 0}^2 \lambda^{(k)}(t,s,x)$. The corresponding three default events $\{ \tau = \tau_0\}$, $\{\tau = \tau_1\}$ and $\{\tau = \tau_2 \}$ are mutually exclusive. When the reference entity defaults ahead of parties $1$ and $2$, i.e. $\tau = \tau_0$, the contract is terminated and party 1 receives $l(\tau_0, X_{\tau_0})$ from party 2 at time $\tau_0$. When either the market participant  or the counterparty  defaults   first, i.e. $\tau < \tau_0$, the amount that the remaining party gets depends on unwinding mechanism at the default time. We adopt the market convention where the MtM value with counterparty risk provision, denoted by $P$, is used to compute the value upon the participant's defaults \citep[see][]{fujii2013derivative,henry2012counterparty}.

Throughout, we use the notations $x^+ = x \, \indic{x \ge 0}$ and $x^- = - x \, \indic{x < 0}$. Suppose that party $2$ defaults first, i.e. $\tau = \tau_2$.  If the MtM value at default is positive $(P_{\tau_2} \ge 0)$, then party $1$ incurs a loss only if the contract is under-collateralized by party $2$ $(\delta_{2} < 1)$ since the amount $\delta_{2} \, P^+_{\tau_2}$ is secured as a collateral. As a result, with the loss rate $L_2$ (i.e. 1 - recovery rate) for party $2$ , the total loss of party $1$ at $\tau_2$ is $L_2 \,(1 - \delta_{2})^+ P^+_{\tau_2}$. On the other hand, suppose that the MtM value is negative $(P_{\tau_2} < 0)$. Party $1$ has a loss only if party $1$ puts collateral more than the MtM value $P_{\tau_2}$,  i.e. the contract is over-collateralized $(\delta_{1} \ge 1)$. In this case, party $1$'s total loss is the product of the party $2$'s loss rate and the exposure, i.e. $L_2 \,(\delta_{1}-1)^+ P^-_{\tau_2}$. Therefore, the remaining value of the party $1$'s position at the default time $\tau_2$ is
\begin{align}
P_{\tau_2} - L_2 \, (1 - \delta_{2})^+ \, P_{\tau_2}^+ - L_2 \, (\delta_{1} - 1)^+ \, P_{\tau_2}^{-}\,. \label{liquidation_value_tau2}
\end{align}

Next, we consider the case when party $1$ defaults first, i.e. $\tau = \tau_1$. We denote by $L_1$ the loss rate of party $1$. If the MtM value of party $1$'s position at the default is negative $(P_{\tau_1} < 0)$ and the contract is under-collateralized $(\delta_{1} < 1)$, party $2$'s loss is $L_1 \,(1 - \delta_{1})^+ P_{\tau_1}^-$.  Similarly, when the MtM value is positive $(P_{\tau_1} \ge 0)$ and the contract is over-collateralized $(\delta_{1} \ge 1)$, party $2$ incurs a loss of the amount $L_1 \,(\delta_{2} - 1)^+ P_{\tau_1}^+$. Because of the bilateral nature of the contract, party $2$'s loss is party $1$'s gain. Therefore, at the default time $\tau_1$, the value of party $1$'s position is
\begin{align}
P_{\tau_1} + L_1 \, (1 - \delta_{1})^+ \, P_{\tau_1}^{-} + L_1 \, (\delta_{2} - 1)^+ \, P_{\tau_1}^+\,. \label{liquidation_value_tau1}
\end{align}

Moreover, the market participant is exposed to funding cost associated with collateralization over the period since the collateral rate and funding rate do not coincide with the risk-free rate. When the liquidation value of the contract $P_t$ is positive to party $1$ at time $t$, party $2$ posts collateral $\delta_{2} \, P_t^+$ to party $1$. To keep the collateral, party $1$ continuously pays collateral interest at rate $c_{2}$ to party $2$ until any default time or expiry. On the other hand, when $P_t$ is negative to party $1$, party $1$ borrows $\delta_{1} \, P_t^-$ to post collateral to party $2$. As a result, party $1$ receives interest payments at rate $c_{1}$ proportional to collateral amount. We call $c_{i}$ the effective collateral rate of party $i$ $(i = 1,2)$, which is the nominal collateral rate minus the funding cost rate of party $i$. The rates $c_{1}$ and $c_{2}$ can be both negative  in practice if the funding costs are high.   Therefore, party $1$ has the following cash flow generated by the collateral and effective collateral rates\,:
 \begin{align}
   \indic{t < \tau} \left( c_{1} \, \delta_{1} \, P_t^- - c_{2} \, \delta_{2} \, P_t^+ \right), \quad 0 \le t \le T \,. \label{collateral_funding_cost}
 \end{align}

The aforementioned cash flow analysis implies that the pre-default MtM value with \emph{counterparty risk} (CR) provision is given by
\begin{align}
	P(t,s,x ) &= \E_{t,s,x} \bigg[ e^{- \int_t^T (r_u + \lambda_u) \, du} \, g(S_T,X_T) +  \int_t^T  e^{- \int_t^u (r_v + \lambda_v) \, dv} \left( h(S_u,X_u) + \lambda^{(0)}_u \, l(t,X_u)   \right) \, du \notag \\
					      & +  \int_t^T \lambda^{(2)}_u \, e^{- \int_t^u (r_v + \lambda_v) \, dv} \left( (1 - L_2 \, (1 - \delta_{2})^+) \, P_u^+ - (1 + L_2 \, (\delta_{1} - 1)^+) \, P_u^{-} \, \right) du  \notag \\
					  & + \int_t^T \lambda^{(1)}_u \, e^{- \int_t^u (r_v + \lambda_v) \, dv} \left( (1 + L_1 \, (\delta_{2} - 1)^+) \, P_u^+ - (1 - L_1 \, (1 - \delta_{1})^+) \, \,  P_u^{-} \, \right) du \nonumber \\
& +  \int_t^T  e^{- \int_t^u (r_v + \lambda_v) \, dv} \left(c_{1} \delta_{1} P_u^- - c_{2}  \delta_{2} P_u^+  \right) \, du \, \bigg]. \label{Expectation_With_Original}
\end{align}
The first line accounts for the terminal cash flow, the dividend, and the payoff at the reference asset's default ($\tau = \tau_0$). The second line and the third line are the cash flows at party 2's default $(\tau = \tau_2)$ in \eqref{liquidation_value_tau2} and party 1's default $(\tau = \tau_1)$ in \eqref{liquidation_value_tau1}, respectively. The last line results from the collateral and effective collateral rates in \eqref{collateral_funding_cost}. To simplify, we introduce the following notations
\begin{align}
\tilde{r}(t,s,x) &= r(t,x) + \lambda(t,s,x)\,,  \label{tilder}\\
 \alpha(t,s,x) &= L_2 \, \lambda^{(2)}(t,s,x) \, (1 - \delta_{2} )^+ -L_1 \, \lambda^{(1)}(t,s,x) \, (\delta_{2} - 1)^+ + c_{2} \,  \delta_{2}\,, \label{alpha}\\
\beta(t,s,x) &=  L_1 \, \lambda^{(1)}(t,s,x) \,(1 - \delta_{1})^+ -  L_2 \, \lambda^{(2)}(t,s,x) \, (\delta_{1} - 1 )^+ + c_{1} \,  \delta_{1}\,, \label{beta} \\
f(t,s,x,y) &= h(s,x) + \lambda^{(0)}(t,s,x) \, l(t,x) + (\lambda^{(1)} + \lambda^{(2)} - \beta)(t,s,x) y +  (\beta - \alpha)(t,s,x) \, y^+\,. \label{def:ftsxy}
\end{align}
This allows to express \eqref{Expectation_With_Original} in the equivalent but simplified form:
\begin{align}
	P(t,s,x )&= \E_{t,s,x} \bigg[ e^{- \int_t^T \tilde{r}_u \, du} \, g(S_T,X_T) +  \int_t^T  e^{- \int_t^u \tilde{r}_v \, dv} f(u,S_u,X_u,P_u) \, du  \bigg]\,, \label{Expectation_With}
\end{align}
where $\tilde{r}_t \equiv \tilde{r}(t,S_t,X_t)$ as defined in \eqref{tilder}. 

\begin{remark} \label{remark:without}
As an alternative of MtM value with CR provision, the liquidation value at the time of default can be evaluated as the CRF value of the claim. In other words, at the default time $\tau < \tau_0$, the liquidation value is evaluated as $\Pi_{\tau}$ rather than $P_{\tau}$. Replacing $P_u$ in \eqref{Expectation_With} with $\Pi_u$ for $t \le u \le T$ gives the MtM value without CR provision (see \cite{henry2012counterparty}):
\begin{align}
	\hP(t,s,x ) & = \E_{t,s,x} \bigg[ e^{- \int_t^T \tilde{r}_u \, du} \, g(S_T,X_T) +  \int_t^T  e^{- \int_t^u \tilde{r}_v \, dv} f(u,S_u,X_u,\Pi_u) \, du  \bigg]. \label{Expectation_Without}
\end{align}
\end{remark}

To conclude this section, we summarize the symbols and their financial meanings in Table \ref{table:defnition} which we will use frequently throughout this paper.
\begin{table}[ht]
\begin{small}
\centering 
\begin{tabular}{c | c | c | c} 
\hline 
Symbol & Definition & Symbol & Definition for party $i \in \{1,2\}$   \\ [0.5ex] 
\hline 
\hline
$P$ & MtM value with CR provision & $R_i$ & Recovery rate \\ 
$\hP$ & MtM value without CR provision  & $c_{i}$ & Effective collateral rate \\
$\Pi$ & CRF value & $\delta_{i}$ & Collateralization ratio\\
$\tau_0$ & Default time of reference asset & $\tau_{i}$ & Default time\\
 $\lambda^{(0)}$ & Default intensity of reference asset & $\lambda^{(i)}$ & Default intensity \\
\hline 
\end{tabular}
\caption{\small{Summary of notations.} \label{table:defnition}} 
\end{small}
\end{table}

\subsection{Bid-Ask Prices} \label{general:bid-ask}
In  OTC trading,   market participants, such as dealers,  may take a long position  as a buyer or a short position as a seller.  Without counterparty risk, the buyer's CRF \textit{bid price} $\Pi^b(t,s,x)$ for a claim with payoff $(g,h,l)$ is  given by  \eqref{Expectation_Pi}.  
The MtM value of the seller's position satisfies \eqref{Expectation_Pi} by replacing $(g,h,l)$ with  $(-g,-h,-l)$, the negative of  which  gives the  seller's CRF \textit{ask price} $\Pi^s(t,s,x)$.  In fact, the bid-ask prices  are identical, i.e. $\Pi^b(t,s,x) = \Pi^s(t,s,x)$.

 Similarly for the case with counterparty risk provision, the buyer's bid price   is  $P^b(t,s,x) = P(t,s,x)$ as  in \eqref{Expectation_With}. The seller's ask price  is given by
\begin{align}
	P^s(t,s,x ) &= \E_{t,s,x} \bigg[ e^{- \int_t^T \tilde{r}_u \, du} \, g(S_T,X_T) +  \int_t^T  e^{- \int_t^u \tilde{r}_v \, dv} \tilde{f}(u,S_u,X_u,P^s_u) \, du \bigg], \label{Expectation_With_seller}
\end{align}
where \begin{align}
\tilde{f}(t,s,x,y) = h(s,x) + \lambda^{(0)}(t,s,x) l(t,x) + (\lambda^{(1)} + \lambda^{(2)} - \beta)(t,s,x) y - (\beta - \alpha)(t,s,x) y^-\,. \label{def:f_tilde}
\end{align}

Since  $\tilde{f}(t,s,x,y)$ is different from $f(t,s,x,y)$ in \eqref{def:ftsxy},  the symmetry observed in the CRF prices generally  no longer holds in the  presence of bilateral counterparty risk.  Most importantly, such an  asymmetry generates  bid-ask spreads for defaultable claims. For any contract with counterparty risk  provision, the participant can quote two  prices: $P^b(t,s,x)$ as a buyer  or  $P^s(t,s,x)$   as a seller. In addition, since the payoff components $(g,h,l)$ can be negative, the bid and/or ask prices also can be negative (see Figure \ref{Bid-Ask_Stock_Digital}).

The total valuation adjustment (XVA) is defined as a deviation of the MtM value from the CRF value, namely, $\Pi - P^b$ for a long position and $P^s - \Pi$ for a short position. The bid-ask spread accounting for the XVA with CR provision is defined as $\Spread(t,s,x) = P^s(t,s,x) - P^b(t,s,x)$.

The two factors $\alpha$ and $\beta$ in \eqref{alpha} and \eqref{beta} that appear in $f$ and $\tilde{f}$ summarize the effects of counterparty risk and collateralization on the bid-ask prices.  Specifically, $\alpha$ explains the effect of positive counterparty exposure of the MtM value $P^+_u$ while $\beta$ explains the effect of negative exposure $P^-_u$. When the two parameters have the same value ($\alpha = \beta$), the two functions $f$ and $\tilde{f}$ in \eqref{def:ftsxy} and \eqref{def:f_tilde} are identical. Therefore, the bid-ask prices  $P^b$ and $P^s$ are   equal. Such a price symmetry also arises in a number of other scenarios: (i) when both parties have perfect collateralization ratio ($\delta_1 = \delta_2 = 1$) and the same effective collateral rate ($c_1 = c_2$); (ii) when both parties have zero collateralization ratio $(\delta_1 = \delta_2 = 0)$ with the same  effective default rate ($L_1 \, \lambda^{(1)}= L_2 \, \lambda^{(2)}$), and  (iii) when both parties have the same effective collateral rate $(c_1 = c_2)$ with the same effective default rate and collateralization ratio ($L_1 \, \lambda^{(1)} = L_2 \, \lambda^{(2)} $, $\delta_1 = \delta_2$). 

\begin{remark} When the counterparty risk-free value $\Pi$ is used to estimate the liquidation value upon default, the seller's bid price is given by
\begin{align}
	\hP^s(t,s,x ) &= \E_{t,s,x} \bigg[ e^{- \int_t^T \tilde{r}_u \, du} \, g(S_T,X_T) +  \int_t^T  e^{- \int_t^u \tilde{r}_v \, dv} \tilde{f}(u,S_u,X_u,\Pi_u) \, du \bigg]\,, \label{Expectation_Without_seller}
\end{align}
where $\tilde{f}$ is defined in \eqref{def:f_tilde}. In contrast to \eqref{Expectation_With}, the price function on the LHS does not appear on the RHS.  \end{remark}


\section{Fixed Point Method} \label{sec:fixed_point}
The defining equation  \eqref{Expectation_With} has a recursive form whereby the price function $P$ appears on both sides.   Denote the spacial domain by $\D : = \R^+ \times \R$. For any  function $w \in C_b([0,T] \times \D, \R)$, we define the operator $\mathcal{M}$ by
\begin{align}
 (\mathcal{M}w)(t,s,x) =  \E_{t,s,x} \bigg[ e^{- \int_t^T \tilde{r}_u \, du} \, g(S_T,X_T) +  \int_t^T  e^{- \int_t^u \tilde{r}_v \, dv} f(u,S_u,X_u,w(u,S_u,X_u)) \, du  \bigg]\,. \label{def:M}
\end{align}
Then, we recognize from  \eqref{Expectation_With} that the  MtM value with counterparty risk provision satisfies $P = \M P$. This motivates  us to show  that the operator  $\mathcal{M}$ has a unique fixed point, and therefore, guarantees the existence and uniqueness of the MtM value $P$.   

We discuss our fixed point approach by first showing that the operator $\mathcal M$  defined in \eqref{def:M} preserves boundedness and continuity. To this end, we outline a number of conditions according to \cite{heath2000martingales}.
\begin{enumerate}[(C1)]
 \newcounter{enumTemp}
\stepcounter{enumTemp}
 \item We define
\begin{align*}
	  \Gamma(t,s,x) = \left[ \begin{array}{c}
r(t,x) + \lambda^{(0)}(t,s,x)  \\
b(t,x) \end{array} \right] \quad \text{and} \quad  \Sigma(t,s,x) &= \left[ \begin{array}{cc}
\sigma(t,s) \, s & 0  \\
\rho \, \eta(t,x) \,  & \sqrt{1 - \rho^2} \, \eta(t,x)   \end{array} \right].
	\end{align*}
The coefficients $\Gamma$ and $\Sigma$  are locally Lipschitz-continuous in $s$ and $x$, uniformly in $t$. That is, for each compact subset $F$ of $\D$, there is a constant $K_F < \infty$ such that for $\psi \in \{ \Gamma, \Sigma \}$,
 \begin{align*}
 	| \psi(t,s_1,x_1) - \psi(t,s_2,x_2)| \le K_F ||(s_1,x_1) - (s_2,x_2)|| \quad \forall t \in [0,T],\, (s_1,x_1)\,,(s_2,x_2) \in F,
 \end{align*}
where $\| \cdot \|$ is the Euclidean norm in $\R^2$.
 \item For all $(t,s,x) \in [0,T) \times \D$, the solution $(S,X)$ neither explodes nor leaves $\D$ before $T$, i.e.
  \begin{align*}
   \P \left( \sup_{t \le u \le T} \| (S_u,X_u) \| < \infty \right) = 1 \quad \text{and} \quad \P \bigg( (S_u,X_u) \in \D \,, \forall u \in [t,T] \bigg) = 1\,.
  \end{align*}
\item The functions $h$ and $g$ are bounded and continuous,  and $r$, $l$ and $\lambda^{(i)}$,\, $i \in \{0,1,2\}$, are positive, continuous and bounded.
  \setcounter{enumTemp}{\theenumi}
\end{enumerate}
\newpage
\begin{lemma} \label{lemma_1} Given any function $w \in C_b([0,T] \times \D, \R)$, it follows that  $v := \M w \in C_b([0,T] \times \D, \R)$.
\end{lemma}
\begin{proof} The boundedness of $v$ follows directly from that of $w$, $h$,  $g$, $r$, and $\lambda^{(i)}$ (see condition $(C3)$). To prove the continuity of $v$, we first observe that 
\begin{align}
	(t,s,x) \mapsto  e^{- \int_t^T \tilde{r}_u \, du} \, g(S_T,X_T) +  \int_t^T  e^{- \int_t^u \tilde{r}_v \, dv} f(u,S_u,X_u,w(u,S_u,X_u)) \, du \label{interim_equation}
\end{align}
is  continuous $\P$-a.s. Indeed,  the continuity of $(S,X)$ implies that the mapping $(t,s,x) \mapsto g(S_T,X_T)$ is  continuous $\P$-a.s. Also, $(t,s,x,u) \mapsto \tilde{r}(u,S_u,X_u)$  and $(t,s,x,u) \mapsto f(u,S_u,X_u,w(u,S_u,X_u))$ are  uniformly continuous and bounded $\P$-a.s. on compact subsets of $[0,T] \!\times \!\D\! \times\! [t,T]$. Hence, the mapping in \eqref{interim_equation} is continuous $\P$-a.s. Recall that $v = \M w$ is the expectation of the RHS  in  \eqref{interim_equation}, which  is  bounded continuous, so $v$ is also continuous for $(t,s,x)$ $\in$ $[0,T] \times \D$ by Dominated Convergence Theorem.
\end{proof}


\subsection{Contraction Mapping} \label{contraction_mapping}
Next, we show that the mapping $\M$ is a contraction. By the boundedness of $\alpha(t,s,x)$, $\beta(t,s,x)$ and $\lambda^{(i)}(t,s,x)$ for $i \in \{0,1,2\}$, we can define a finite positive  constant  by
\begin{align*}
L &= \sup_{(t,s,x) \in [0,T] \times \D} \left\{|\lambda^{(1)}(t,s,x) + \lambda^{(2)}(t,s,x) - \beta(t,s,x)| + |\beta(t,s,x)- \alpha(t,s,x)| \right\} .
\end{align*} 
\begin{proposition} \label{prop:contraction_mapping}
The mapping $\mathcal{M}$ defined in \eqref{def:M} is a contraction on the space $C_b([0,T] \times \D, \R)$ with respect to the norm
\begin{align}
 \| w \|_{\gamma} := \sup_{(t,s,x) \in [0,T] \times \D} e^{- \gamma \, (T-t)} |w(t,s,x)| \,,\label{def:norm_beta}
\end{align}
for $L < \gamma < \infty$. In particular, $\mathcal{M}$ has a unique fixed point   $w^* \in C_b([0,T] \times \D, \R)$.
\end{proposition}
\begin{proof} From \eqref{def:ftsxy}, we observe that $ | f(t,s,x,y_1) - f(t,s,x,y_2)|\le L \, | y_1 - y_2 |$,  for $(t,s,x) \in [0,T]\! \times\! \D$. This implies $f$ is \textit{Lipschitz-continuous} in $y$, uniformly over $(t,s,x)$. By Lemma \ref{lemma_1}, the operator $\M$ maps $C_b([0,T] \times \D, \R)$ into itself. For   $(t,s,x) \in [0,T] \times \D$, \,  $w_1, w_2 \in C_b([0,T] \times \D, \R)$, and $\gamma >0$, we have
\begin{align}
 &e^{- \gamma \, (T-t)} | (\mathcal{M} w_1)(t,s,x) - (\mathcal{M} w_2)(t,s,x)| \notag \\
 &= e^{- \gamma \, (T-t)} \bigg\vert \E_{t,s,x} \left[ \int_t^T  e^{- \int_t^u \tilde{r}_v \, dv} ( f(u,S_u,X_u,w_1(u,S_u,X_u) - f(u,S_u,X_u,w_2(u,S_u,X_u)  ) ) \, du  \right]\bigg\vert \notag \\
 &\overset{(i)}\le e^{- \gamma \, (T-t)} \E_{t,s,x} \left[ \int_t^T   \bigg\vert  f(u,S_u,X_u,w_1(u,S_u,X_u) - f(u,S_u,X_u,w_2(u,S_u,X_u))  \bigg\vert \, du  \right]\notag \\
  &\overset{(ii)}\le e^{- \gamma \, (T-t)} \E_{t,s,x} \left[ \int_t^T  e^{-\gamma (T-u)}   L \vert w_1(u,S_u,X_u) - w_2(u,S_u,X_u)\vert  e^{\gamma (T-u)}  \, du  \right]\notag \\
 & \overset{(iii)}\le e^{- \gamma \, (T-t)}  L \| w_1 - w_2 \|_{\gamma} \int_t^T e^{  \gamma \, (T-u)} du \notag \\
 &\le \frac{L}{\gamma} \| w_1 - w_2 \|_{\gamma}\,. \notag
 \end{align}
We have used the condition that $\tilde{r}_v \ge 0$ for $(i)$, and the fact that $f$ is Lipschitz in $y$ for $(ii)$. The inequality $(iii)$ is implied by the  norm   in \eqref{def:norm_beta}. As a result, for any $\gamma > L\ge 0$,  $\mathcal{M}$ is a contraction. 
\end{proof}
The norm $\| \cdot \|_{\gamma}$   is equivalent to the   supremum norm $\|\cdot\|_{\infty}$ on the space $C_b([0,T] \times \D, \R)$. A similar norm is used in \cite{becherer2005classical} and \cite{leung2009accounting} in their studies of reaction diffusion PDEs  arising from indifference pricing.

Using the fact that $\M$ is a contraction proved in Proposition \ref{prop:contraction_mapping}, there exists a sequence of functions $(P^{(n)} )_{n \ge 0}$ that satisfy $P^{(n+1)} = \M P^{(n)}$, $\forall n \ge 0$, and the sequence converges to the fixed point $P$. The convergence does not rely on the choice of the initial function. Indeed, one can simply pick any bounded continuous function as a starting point, e.g. $P^{(0)} = 0$ \,$\forall (t,s,x)$, and iterate to have a sequence $(P^{(n)})_{n \ge 0}$ that resides in $C_b([0,T] \times \D, \R)$. 

Furthermore, we can  show  that for each $n \ge 1$, $P^{(n)} \equiv P^{(n)}(t,s,x)$ is a classical solution of the following inhomogeneous PDE problem:
\begin{align}
  \frac{\partial P^{(n)}}{\partial t} + \L P^{(n)} - \tilde{r}(t,s,x) \, P^{(n)}  + f(t,s,x,P^{(n-1)}) &= 0\,, \notag\\
  P^{(n)}(T,s,x) &= g(s,x)\,, \label{PDE_With_iteration} 
\end{align}
where the operator $\L$ is defined by
\begin{align}
\L &:=  \frac{1}{2} \sigma(t,s)^2 \, s^2 \, \frac{\partial^2}{\partial s^2} + \frac{1}{2} \eta(t,x)^2 \frac{\partial^2 }{\partial x^2} + \rho \, \eta(t,x) \, \sigma(t,s) \, s \, \frac{\partial^2}{\partial s \partial x} \notag \\
& \quad + \tilde{r}(t,s,x) \,s \,  \frac{\partial}{\partial s} + b(t,x) \frac{\partial }{\partial x}\,. \label{operator_L_Equity_Credit}
\end{align}
In order to prove the result, we need the following additional conditions, adapted in our notation from $(A3') - (A3d')$ of \cite{heath2000martingales}.
\begin{enumerate}[(C1)]
\setcounter{enumi}{\theenumTemp}
	\item There exists a sequence $(D_n)_{n \in \N}$ of bounded domains with closure $\bar{D}_n \subset D$ such that $\cup_{n=1}^{\infty} D_n = D$ and each $D_n$ has a $C^2$-boundary.
	  \setcounter{enumTemp}{\theenumi}
	\end{enumerate}		
As in \cite{heath2000martingales}, one can take $\mathcal{D}_n=[\frac{1}{n},n]\!\times\! [-n,n] \subset \R_+\!\times\! \R$.	For each $n$,  we require that
	\begin{enumerate} [(C1)]
	\setcounter{enumi}{\theenumTemp}
	\item $b(t,x)$, $a(t,s,x) := \Sigma(t,s,x) \, \Sigma^{t}(t,s,x)$, and $\tilde{r}(t,s,x)$ be uniformly Lipschitz-continuous on $[0,T] \times \bar{D}_n$, where $\Sigma^{t}$ denotes the transpose matrix of $\Sigma$, 
	\item $a(t,s,x)$ be uniformly elliptic on $\R^2$ for $(t,s,x) \in [0,T) \times D_n$, i.e. there is $\delta_n > 0$ such that $y^{t} \, a(t,s,x) \, y \ge \delta_n \|y\|^2$ for all $y \in \R^2$,
	\item $f(t,s,x,y)$ be uniformly H\"{o}lder-continuous on $[0,T] \times \bar{D}_n \times \R$.
\end{enumerate}
The conditions (C1) -- (C7) are quite  general, and they allow for various models, including  the Heston, CEV, and thus,  geometric Brownian motion  models for equity, and the Ornstein-Uhlenbeck and Cox-Ingersoll-Ross models for the stochastic factor $X$ \citep[Sect. 2]{heath2000martingales}. The triplet $(g,h,l)$, default intensities $\lambda^{(i)}$ and interest rate $r$ can be easily chosen to satisfy the boundedness and continuity conditions in (C3), as we will do in our examples in Sections \ref{Equity_Claim} and  \ref{sec:BCVA_Credit}.

\begin{theorem} \label{prop:PDE_classical_solution}
 Under conditions $(C1) - (C7)$, there exists a sequence of bounded classical solutions $(P^{(n)})\subset C^{1,2}_b([0,T) \times \D, \R)$ of the PDE problem \eqref{PDE_With_iteration} that  converges to the fixed point $P \in C_b([0,T) \times \D, \R)$ of the operator $\M$.
\end{theorem}

We provide the proof in Appendix \ref{sect-proofthm}. The insight of Proposition \ref{prop:contraction_mapping} and Theorem \ref{prop:PDE_classical_solution} is that we can construct and solve a series of inhomogeneous but  {linear} PDEs whose classical solutions converge  to  a unique fixed point price function   $P$ as in  \eqref{Expectation_With}. Recent studies by   \cite{burgard2011partial} and \cite{henry2012counterparty}  evaluate the MtM value $P(t,s,x)$ by working  with  the associated   nonlinear  PDE of the form:
\begin{align}
	 \frac{\partial P}{\partial t} + \L P  - \tilde{r}(t,s,x)  \, P  +  f(t,s,x,P) = 0\,, \label{PDE_With} 
\end{align}
for $(t,s,x) \in [0,T) \times \D$, with  terminal condition  $P(T,s,x) =  g(s,x)$, for $(s,x) \in \D$. The nonlinearity of \eqref{PDE_With}  poses major  challenges  on analyzing and numerically solving  for $P$.  \cite{henry2012counterparty} provides a method to  approximate the solution  that involves replacing  the nonlinear term $f$ with a polynomial and  simulating a  marked branching diffusion. {This method, however, does not guarantee that the solution from simulation will resemble   the solution of the nonlinear PDE, and does not ensure any regularity, such as continuity or boundedness of, either solution.   \cite{henry2012counterparty} provides  conditions on the chosen polynomial  to avoid a ``blow-up" of the simulation algorithm.  In contrast, our fixed point methodology circumvents this issue  by  establishing   that the  pricing definition   in \eqref{Expectation_With}  is  a contraction mapping, as opposed to  working with the nonlinear PDE. As a result, we   solve a series of linear PDE problems  with bounded classical solutions. In the limit, a unique bounded continuous   MtM value $P$ is obtained. 
 
\subsection{Numerical Implementation} \label{numerical_solution}
Our contraction mapping methodology lends itself to a recursive numerical algorithm. As mentioned in the previous section, we iteratively solve a sequence of linear inhomogeneous PDEs \eqref{PDE_With_iteration}. At each iteration, the error is measured in terms of the  maximum difference between two consecutive solutions $P^{(n)}$ and $P^{(n-1)}$ over the entire domain $[0,T] \times \D$. We continue the iteration procedure until the error is less than the pre-defined tolerance level  $\bar{\epsilon}$.

For implementation, we use the standard  Crank-Nicolson finite difference method (FDM) to obtain the values (see, among others, \cite{wilmott1995mathematics} and  \cite{strikwerda2007finite}). We restrict the domain $[0,T] \times \D $ to a finite domain $\bar{\D} = \{ (t,s,x) \,:\, 0 \le t \le T,\, \underline{X} \le x \le \bar{X}, 0 \le s \le \bar{S}\}$. The parameters $\bar{S}$, $\underline{X}$ and $\bar{X}$ are sufficiently large enough to preserve the accuracy of the numerical solutions. We discretize the function $P^{(n)}(t,s,x)$ as $P^{(n)}(t_i,s_j,x_k)$ where $i \in \{0,...,N\}$, $j \in \{0,...,M\}$ and $k \in \{0,...,L\}$ with $\Delta t =  T / N$, $\Delta s = \bar{S} / M$, $\Delta x = (\bar{X} - \underline{X}) / L$ and $t_i = i \, \Delta t$, $s_j = j \, \Delta s$, $x_k = k \, \Delta x$. Our numerical procedure is summarized in Algorithm \ref{Algorithm:Fixed_Point}.

\begin{algorithm}
        \caption{\small{Fixed Point Algorithm for Evaluating the MtM Value $P$}}
        \begin{algorithmic} \label{Algorithm:Fixed_Point}
            \STATE set $n = 1$,\, $P^+ = P^{(0)}$
            \STATE solve for $P^{(1)}$ from PDE \eqref{PDE_With_iteration}
            \STATE set $\epsilon = \| P^{(1)} - P^{(0)} \|_{\infty}$
            \WHILE{$ \epsilon > \bar{\epsilon}$}
            \STATE set $n = n+1$, \, $P^+ = P^{(n-1)}$
            \STATE solve for $P^{(n)}$ from PDE \eqref{PDE_With_iteration}
			\STATE set $\epsilon = \| P^{(n)} - P^{(n-1)} \|_{\infty}$
            \ENDWHILE
            \RETURN $P^{(n)}$
        \end{algorithmic}
    \end{algorithm}

For the CRF value, we solve the linear PDE associated with $\Pi \equiv \Pi(t,s,x)$ in \eqref{Expectation_Pi}, namely,
\begin{align}
	\frac{\partial \Pi}{\partial t} + \L \, \Pi - \tilde{r}(t,s,x) \, \Pi(t,s,x) + h(s,x) + \lambda^{(0)}(t,s,x) \, l(t,x) = 0\,,     \label{PDE_RF}
\end{align}
 for $(t,s,x) \in [0,T) \times \D$, with terminal condition $\Pi(T,s,x) =  g(s,x)$, for $(s,x) \in \D$. The CRF value becomes an input to the PDE problem for the MtM value without provision, given by 
\begin{align}
	\frac{\partial \hP}{\partial t} + \L \hP - \tilde{r}(t,s,x) \, \hP + f(t,s,x,\Pi(t,s,x)) = 0 \,, \label{PDE_Without}
\end{align}
for $(t,s,x) \in [0,T) \times \D$, and   $\hP(T,s,x) =  g(s,x)$, for $(s,x) \in \D$. Again, we apply the Crank-Nicolson FDM method to compute their values.

\begin{remark} The assumption on the boundedness of the default intensities does not encapsulate the local intensity model, where the  reference default intensity function is  of the form: $\lambda(t,s) =  c s^{-p}$, for some $p,c>0$. See \cite{carr2006jump,carr2010local,linetsky2006pricing,madan1998pricing}, among others. As a close alternative, one can cap the exploding intensity and set $\lambda(t,s) =  c s^{-p}\wedge M$ for  some arbitrarily chosen  large constant $M$.  With this modification, our  model and contract mapping results still apply. Since the default intensity function is finite except at $s=0$,   in numerical implementation using finite difference, one can  set the default intensity at the layer $s=0$ to be the large value $M$ instead of an  infinite value.
\end{remark}

\section{Defaultable Equity Derivatives with Counterparty Risk} \label{Equity_Claim}
We now apply our valuation methodology to value  a number of   defaultable equity claims. Specifically, we will  derive and compare the  MtM values with and without counterparty risk provisions as well as the  CRF value. Moreover, we will analyze and illustrate  the  bid-ask prices.    

As a special case of \eqref{model:stock}, we model the pre-default stock price process by 
\begin{align}
 dS_t  &= \left(r +  \lambda^{(0)}\right) \, S_t \, dt + \sigma \, S_t \,  dW_t\,,\label{Sequity}
\end{align}
where we assume constant interest rate $r$ and default rates $\lambda^{(i)}$, $i \in \{0,1,2\}$. In addition, we  let $\lambda = \sum_{k=0}^2 \lambda^{(k)}$, and set 
\begin{align}
\alpha &= L_2 \, \lambda^{(2)} (1 - \delta_{2} )^+ - L_1 \, \lambda^{(1)} (\delta_{2} - 1)^+ + c_{2} \,  \delta_{2}\,, \label{alpha_constant}\\
\beta &= L_1 \, \lambda^{(1)} (1 - \delta_{1})^+ - L_2 \, \lambda^{(2)} (\delta_{1} - 1 )^+ + c_{1} \,  \delta_{1}\,. \label{beta_constant}
\end{align}

\subsection{Call Spreads}
Let us consider a generic call spread with the terminal payoff\,:
\begin{align}
	g(S_T) &=
	\begin{cases}
	   	m_2 & \quad \text{if} \quad  S_T > K + \epsilon_2\,,\\
	   \frac{(m_1 + m_2)}{\epsilon_1 + \epsilon_2} (S_T - K)    &\quad \text{if}\quad  K - \epsilon_1 \le S_T \le K + \epsilon_2\,, \\
	   -m_1 &\quad \text{if} \quad S_T < K - \epsilon_1\,,
	\end{cases} \label{terminal_g_S}
\end{align}
with $m_1, m_2, \epsilon_1, \epsilon_2 > 0$, where $m_1 / \epsilon_1 = m_2 / \epsilon_2 =: M$. The payoff resembles that of a long position of $M$ call options with strike $K - \epsilon_1$,  a short position of $M$ call options with strike $K + \epsilon_2$ and short $m_1$ notional of zero coupon bond with the same maturity. Similar positions can be achieved when two OTC traders buy and sell call options with different strikes, plus/minus some cash. As $\epsilon_1$ and $\epsilon_2$ in \eqref{terminal_g_S} go to zero, the payoff converges to that of a digital call position covered in \cite{henry2012counterparty}.

With the terminal payoff $g$ in \eqref{terminal_g_S},  dividend $h= 0$, and value at   reference default   $l(\tau_0) = -m_1 \, e^{-r (T-\tau_0)}$, the CRF value of the spread contract admits the formula
\begin{align*}
\Pi(t,s) = &M \, \left( C^{BS}(t, s \,; T , K - \epsilon_1,r + \lambda^{(0)},\sigma) -  C^{BS}(t, s \, ; T , K + \epsilon_2,r + \lambda^{(0)},\sigma) \right) - e^{-r (T-t)} \, m_1\,, 
\end{align*}
where $C^{BS}(t,s\,;T,K,r,\sigma)$ is the Black-Scholes call option price at time $t$ with spot price $s$, maturity $T$, strike price $K$, risk-free rate $r$ and volatility $\sigma$. From \eqref{Expectation_With}, the MtM value with counterparty risk provision is given by
\begin{align}
	P(t,s) &= \E_{t,s} \bigg[ e^{- (r + \lambda) (T-t)} \, g(S_T) +  \int_t^T  e^{- (r + \lambda) (u-t)} \, f(u,S_u,P_u) \, du   \bigg]\,, \label{Expectation_With_Digital}
\end{align}
where $f(t,s,y) := \lambda^{(0)} l(t) + (\lambda^{(1)} + \lambda^{(2)} - \beta)  y + (\beta - \alpha) y^+$.
The MtM value without counterparty risk provision $\hP(t,s)$ is similarly obtained replacing $P_u$ in \eqref{Expectation_With_Digital} with $\Pi_u$. 

The   model for $S$ in  \eqref{Sequity},   the triple  $(g,h,l)$,  and other (constant) coefficients satisfy the conditions (C1)-(C7) with domain $\mathcal{D} = \R_+$ (see also  \citep[Sect. 2]{heath2000martingales}). 
We numerically compute the MtM value $P(t,s)$ by Algorithm 1 from Section \ref{numerical_solution}. For the iterative PDE \eqref{PDE_With_iteration}, we adopt the coefficients in this section and the terminal payoff $g(S_T)$ given in \eqref{terminal_g_S}. In Table \ref{table:convergence_digital}, we show the convergence of the MtM values with provision for three different contracts where $\epsilon_1 = \epsilon_2 = \{2,1,0.01\}$. The first column of each contract shows the value of the MtM value of the contract at spot $s= 10$ for each step $0 \le n \le 5$. The second column of each contract shows the supremum norm $\epsilon = \| P^{(n)} - P^{(n-1)}\|_{\infty}$ for each step $0 \le n \le 5$. The algorithm stops at $n = 5$ for all three contracts.
\begin{table}[ht]
\begin{small}
  \centering
    \begin{tabular}{c|cc|cc|cc}
    \hline
         & \multicolumn{2}{c|}{$\epsilon_1 = \epsilon_2 = 2$} & \multicolumn{2}{c|}{$\epsilon_1 = \epsilon_2 = 1$}      & \multicolumn{2}{c}{$\epsilon_1 = \epsilon_2 = 0.01$} \\
    \hline
    \hline
          & $P^{(n)}(0,10)$ & $\epsilon$ & $P^{(n)}(0,10)$ & $\epsilon$  & $P^{(n)}(0,10)$ & $\epsilon$ \\
    $n=0$   & 0    & - & 0    & -     & 0    & - \\
    $n=1$   & -0.1197 & 0.9048 & -0.1293 & 0.9048 & -0.1326 & 0.9048 \\
    $n=2$   & -0.1387 & 0.0992 & -0.1490 & 0.0992 &  -0.1526 & 0.0992 \\
    $n=3$   & -0.1377 & 0.0060 & -0.1479 & 0.0060 & -0.1515 & 0.0060 \\
    $n=4$   & -0.1377 & 0.0002 & -0.1480 & 0.0002 & -0.1516 & 0.0002 \\
    $n=5$   & -0.1377 & $< 10^{-5}$ & -0.1480 & $< 10^{-5}$ & -0.1516 & $< 10^{-5}$ \\
    \hline
    \end{tabular}%
      \caption{\small{Convergence of the MtM values with provision $P(0,s)$ of call spread contract at spot price $s = 10$ (at-the-money) and $m_1 = m_2 = 1$. Parameters: $K = 10$, $T=2$, $t = 0$, $r = 2\%$, $\sigma = 25\%$, $\lambda^{(0)} = 3\%$, $\lambda^{(1)} = 5\%$, $\lambda^{(2)} = 15\%$, $R_1 = 40\%$, $R_2 = 40\%$, $\delta_{1} = \delta_{2} = 0$, $\bar{\epsilon} = 10^{-5}$, $\bar{S} = 40$, $\Delta S = 0.01$, $\Delta t = 1/1000$. }}
\label{table:convergence_digital} 
\end{small}
\end{table} 

Let us  visualize the convergence of the MtM value with CR provision $P^{(n)}(0,s)$ in Figure \ref{BCVA_Digital} (left). Using the tolerance level $\bar{\epsilon} = 10^{-5}$  for the maximum difference over each iteration,  the  algorithm stops after $4$ iterations.  As we can see, the price functions $P^{(3)}(0,s)$ and $P^{(4)}(0,s)$ over $0 \le s \le \bar{S} = 40$ are not visibly distinguishable.

\begin{figure}[h!]
    \centering
    \includegraphics[width=0.45\textwidth]{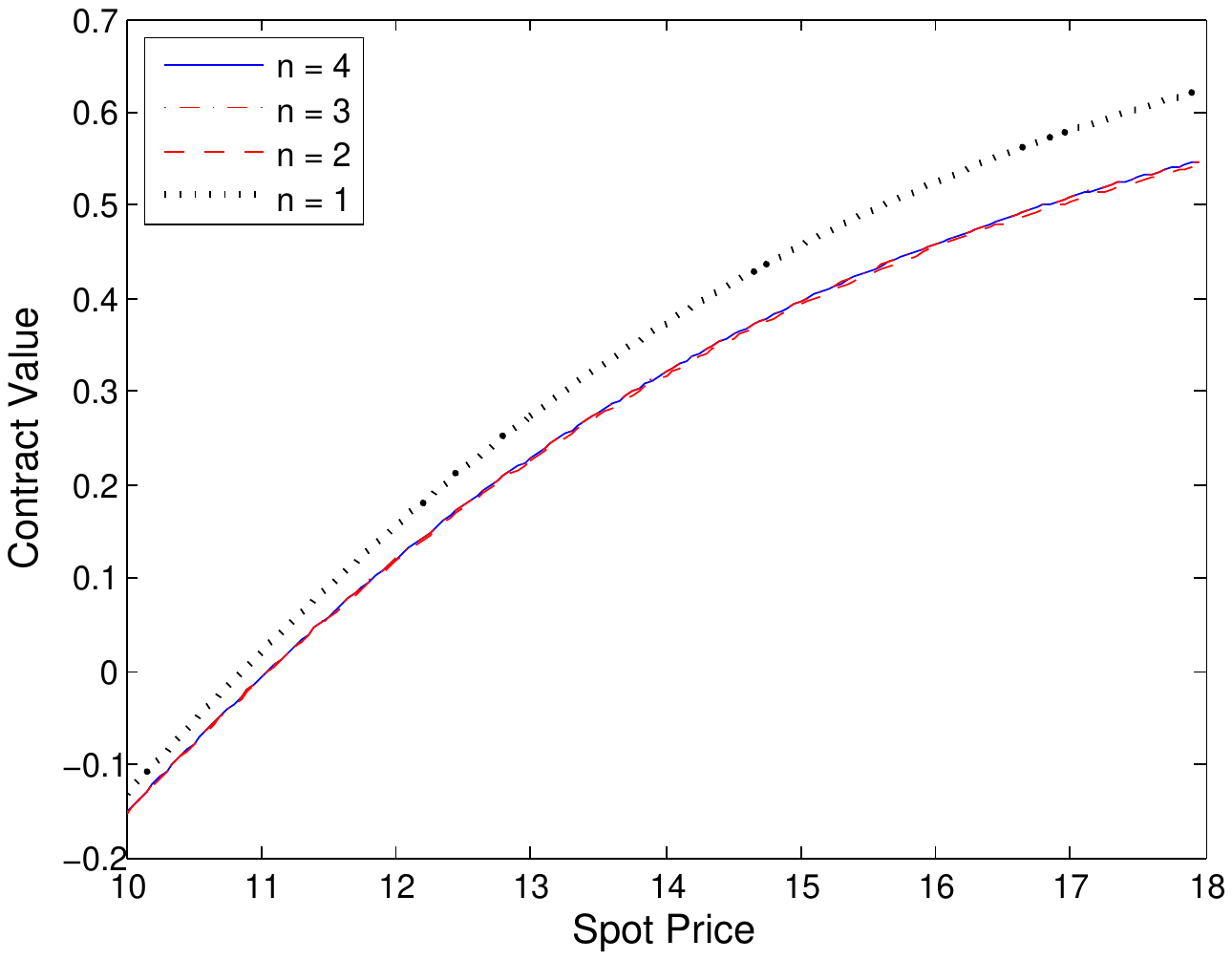}
    \includegraphics[width=0.45\textwidth]{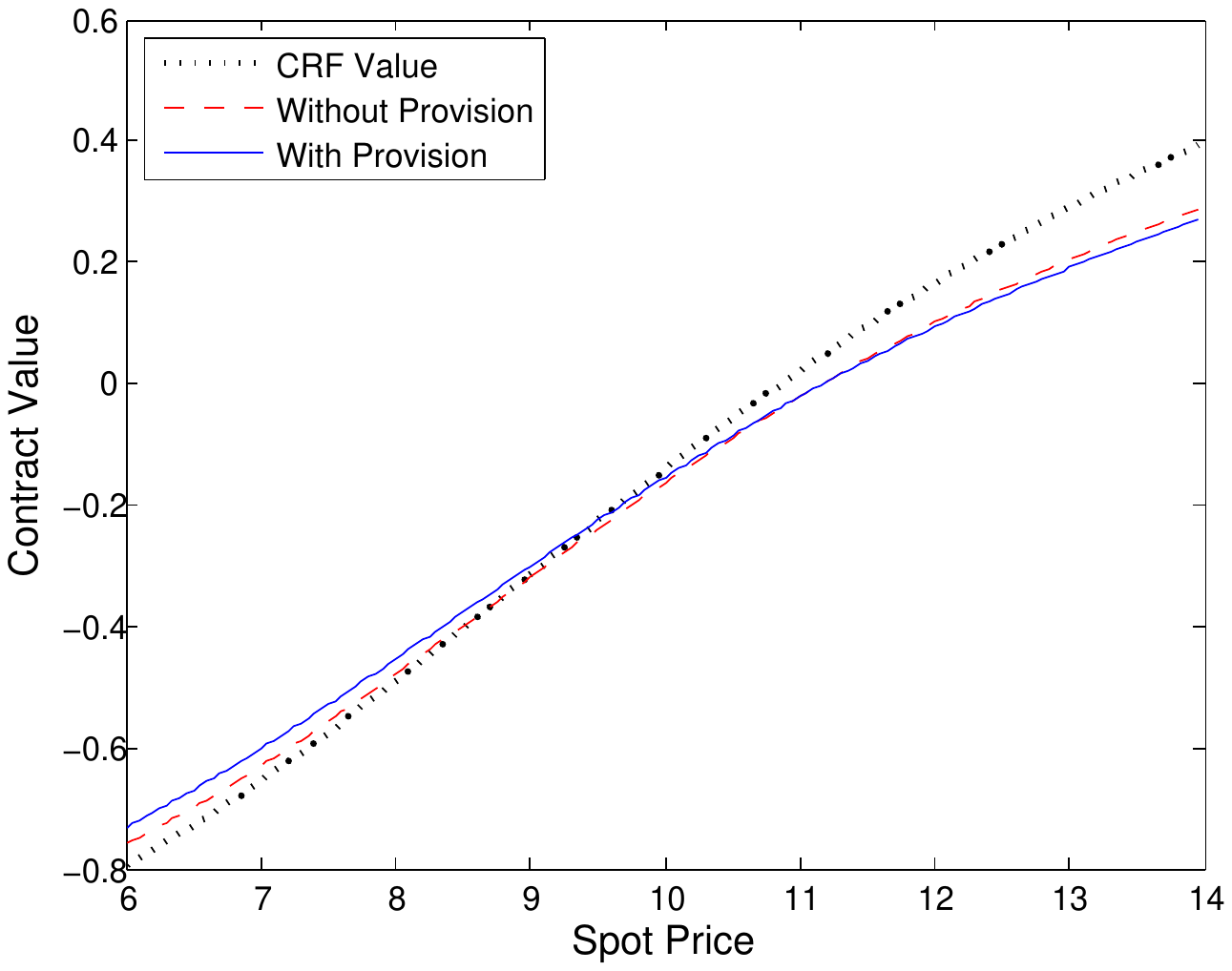}
	 \caption{\small{(Left) Convergence of the MtM values with provision of a call spread $P(0,s)$  for  $s \in [10,18]$. (Right) Comparison of the three MtM values of  a call spread $\{\Pi(0,s), \hP(0,s),P(0,s)\}$ over the spot price. Parameters are given in Table \ref{table:convergence_digital}.}}
    \label{BCVA_Digital}
\end{figure}

In Figure \ref{BCVA_Digital} (right),  we plot three different values $\Pi(0,s)$, $\hP(0,s)$ and $P(0,s)$ for $0 \le s \le \bar{S}$. As we can see, the ordering of these three values can change completely depending on the spot price. For example, for large spot prices, we observe that the CRF value dominates the other two MtM values, but it is lowest when the spot price is small. Furthermore, the value without provision dominates the value with provision for high spot prices, and the opposite holds true for low spot prices.  

Next, we look at the sensitivity of the MtM values with respect to the  counterparty's or own default risk, collateralization ratio and effective collateral rate. In Figure \ref{BCVA_Digital_CounterpartyRisk_Default_Rate}, the MtM values are decreasing in the counterparty default rate (left) and   increasing in the participant's own default rate (right), as is intuitive. Note that the MtM value with provision moves more rapidly with respect to the counterparty default rate, but the MtM value without provision is more sensitive in the participant's own default rate.

In Figure \ref{BCVA_Digital_Collateral_Funding} (left), an increase in  $\delta_{2}$ reduces counterparty-risk exposure, and therefore, increases the MtM values with and without counterparty risk provision. The rate of increase in contract value  slows down when the collateralization ratio exceed 1. In the over-collateralized range $[1, 1.2]$, party 1 is no longer exposed to the counterparty's default risk. The increase in the contact value (from party 1's perspective) results from  the possibility of collecting  the excess collateral upon party 1's own  default.

In practice, if the participant's funding cost rate is high, the effective collateral rate can be negative (see \cite{burgard2011partial}). This implies a net interest payment by the participant for the long position due to collateralization. As the effective collateral rate becomes more negative, the contract values with and without provision decrease as we observe on the right panel of Figure \ref{BCVA_Digital_Collateral_Funding}.

\begin{figure}[h]
    \centering
    \includegraphics[width=0.45\textwidth]{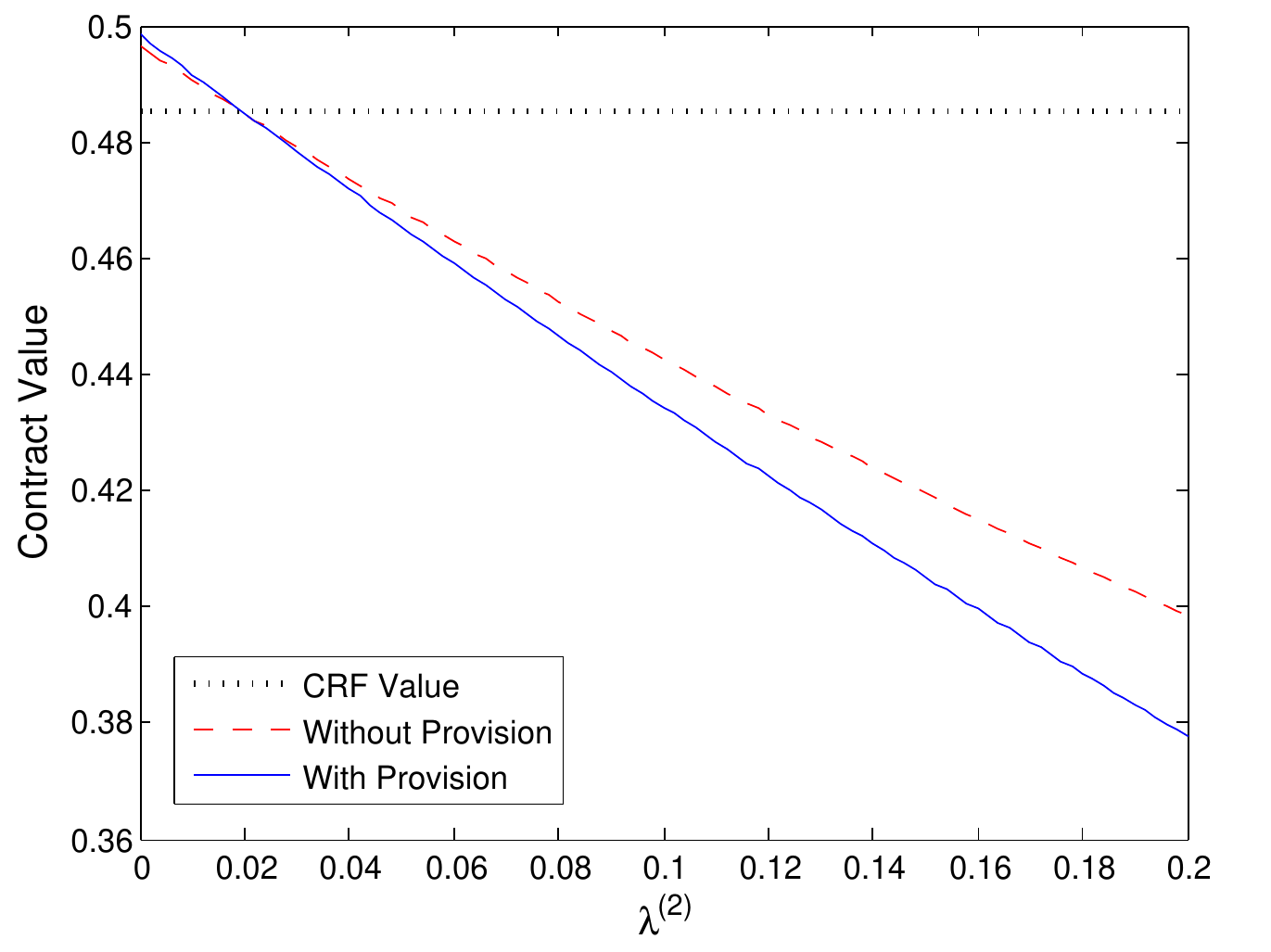}
        \includegraphics[width=0.45\textwidth]{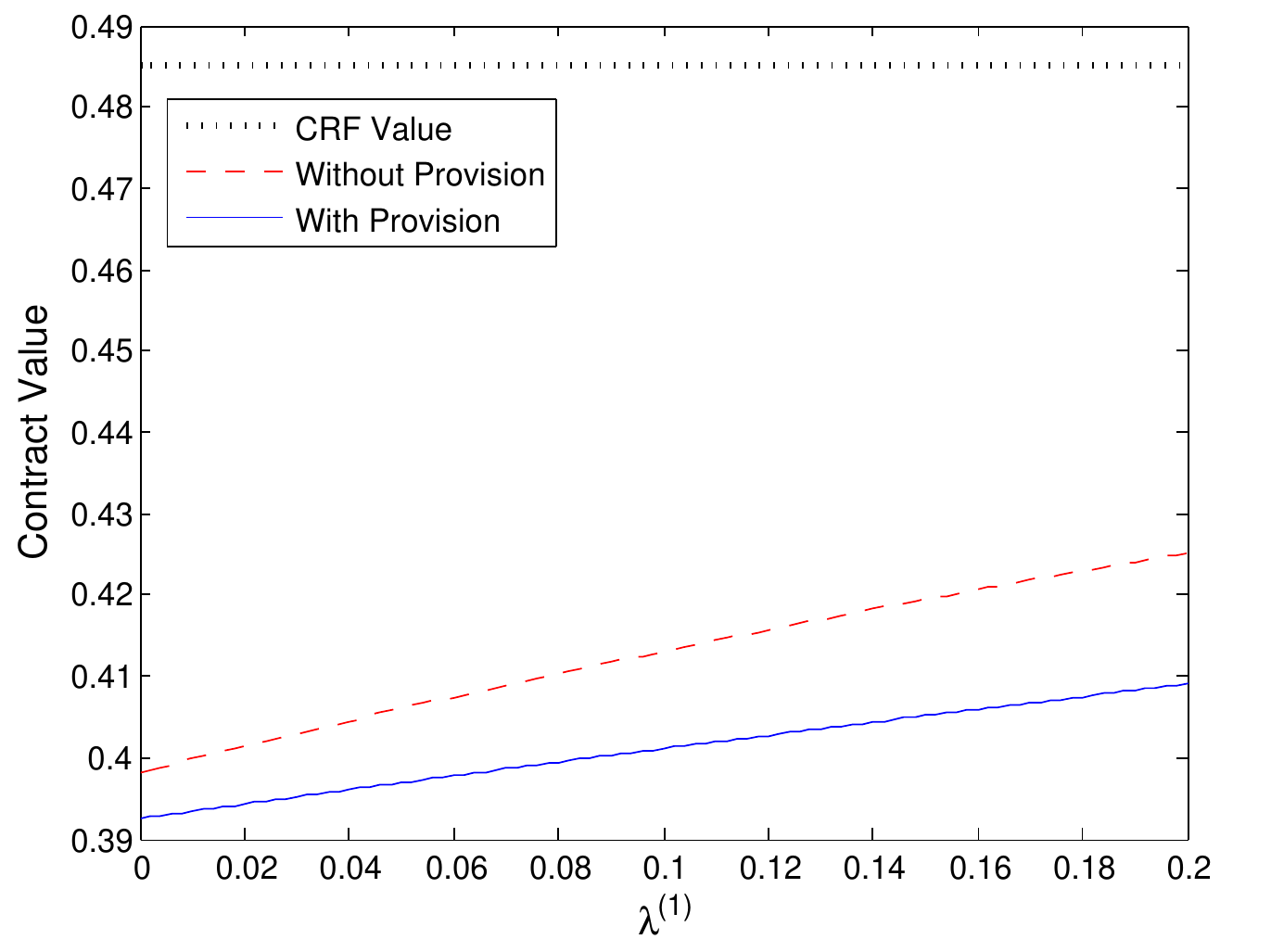}
	 \caption{\small{(Left) The MtM values with and without provision are decreasing in the counterparty default rate $\lambda^{(2)}$ with  $\lambda^{(1)} = 15\%$. (Right) The MtM values are increasing in $\lambda^{(1)}$ with $\lambda^{(2)} = 15\%$. The CRF value stays  constant as  $\lambda^{(2)}$ or  $\lambda^{(1)}$ varies. Parameters: $\epsilon_1 = \epsilon_2 = 0.01$, $m_1 = m_2 = 1$, $s = 15$, $K = 10$, $T = 2$, $t=0$,  $r = 2\%$, $\sigma = 25\%$, $\lambda^{(0)} = 3\%$, $R_1 = R_2 = 40\%$, $\delta_{1} = \delta_{2} = 0\%$, $c_{1} = c_{2} = 1\%$, $\bar{\epsilon} = 10^{-5}$, $\Delta S = 0.05$, $\Delta t = 1/250$.}}
    \label{BCVA_Digital_CounterpartyRisk_Default_Rate}
\end{figure}
\begin{figure}[h!]
    \centering
    \includegraphics[width=0.45\textwidth]{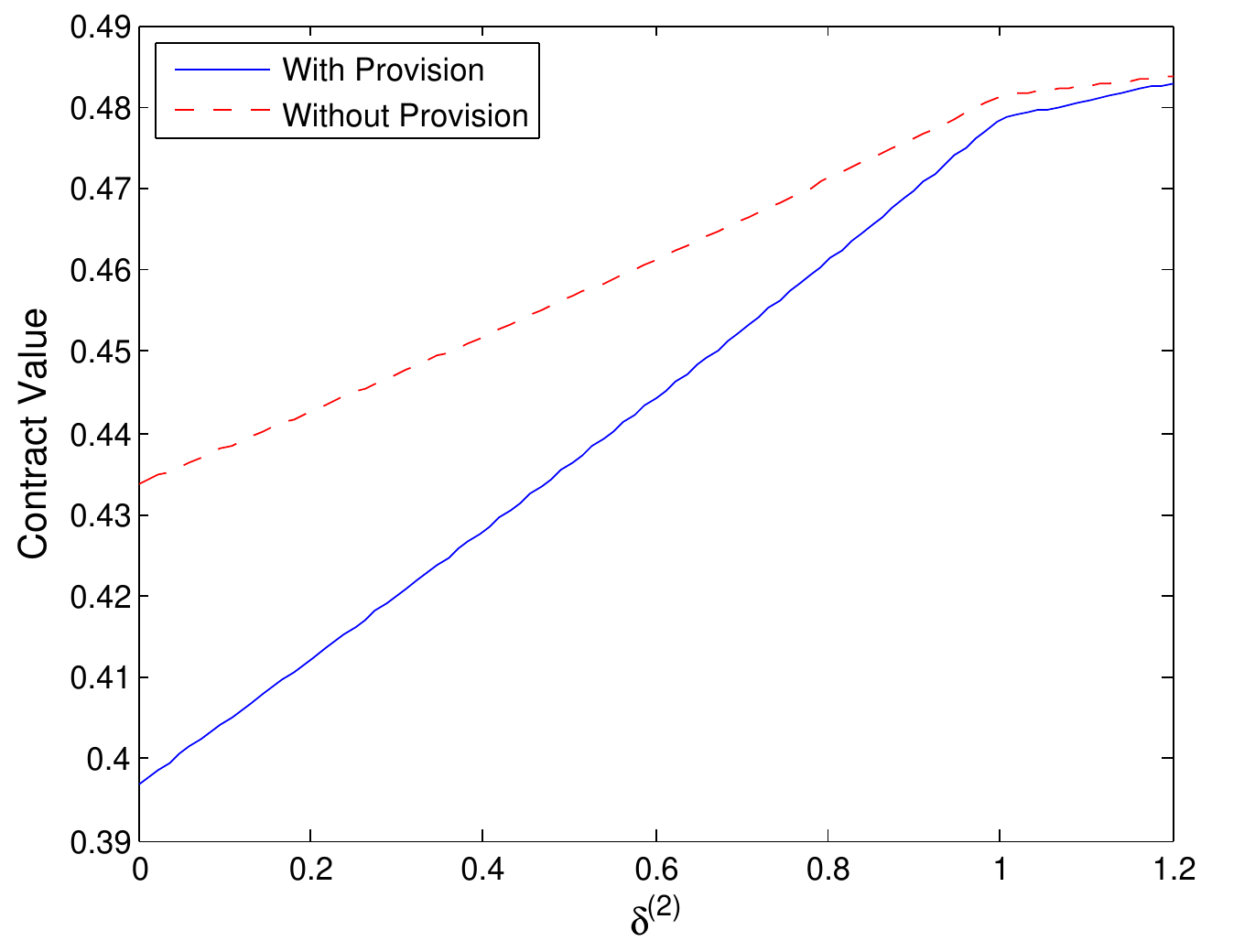}
        \includegraphics[width=0.45\textwidth]{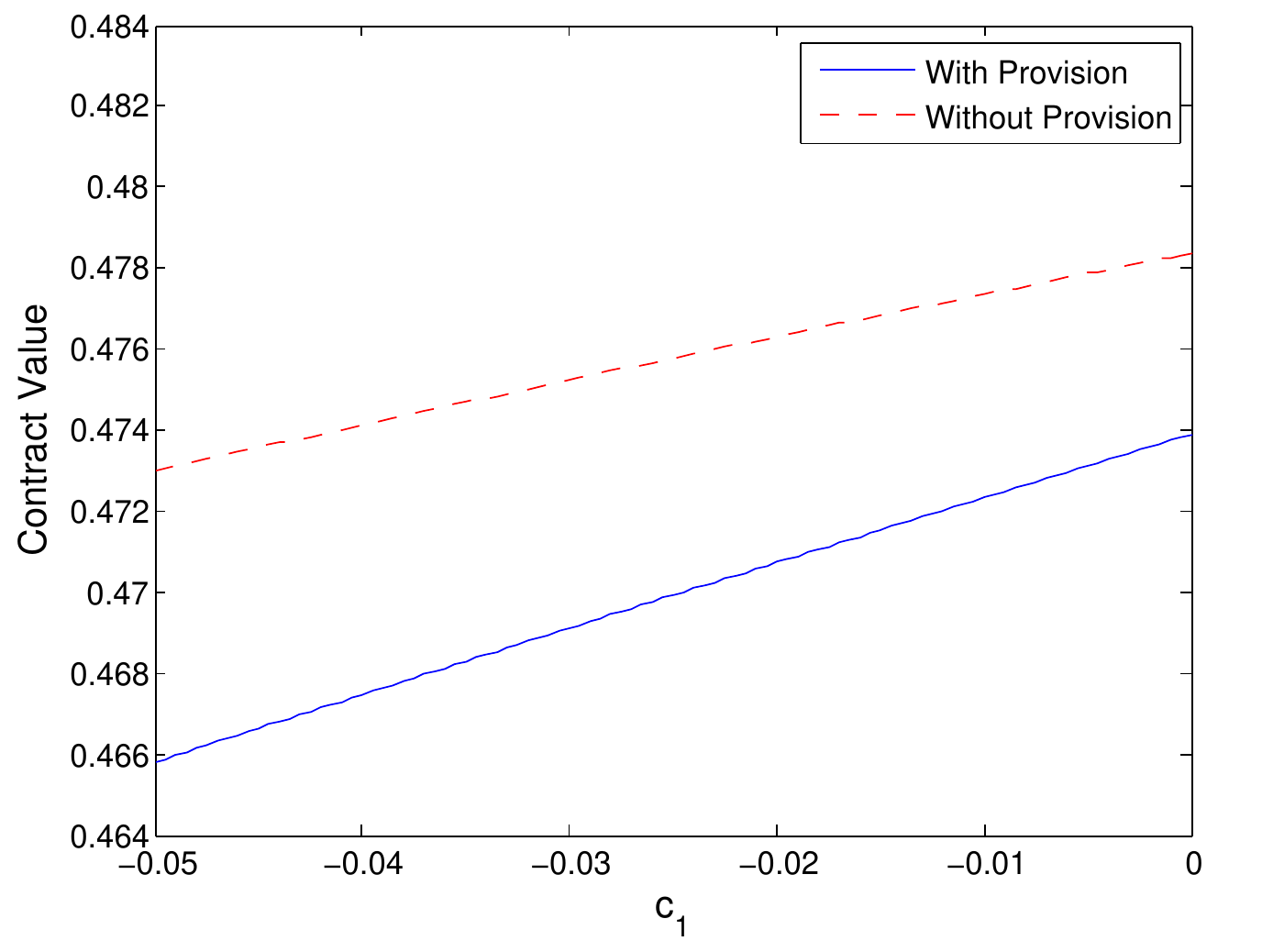}
	 \caption{\small{(Left) The MtM values of a call spread with and without provision are decreasing in counterparty's collateralization ratio $\delta_{2}$. (Right) The MtM values  of a call spread are increasing in the participant's effective collateral rate $c_1 \in [-5\%,0\%]$. Parameters: $\epsilon_1 = \epsilon_2 = 0.01$, $m_1 = m_2 = 1$, $s = 15$, $K = 10$, $T = 2$, $t=0$,  $r = 2\%$, $\sigma = 25\%$, $\lambda^{(0)} = 3\%$, $\lambda^{(1)} = 5\%$, $\lambda^{(2)} = 15\%$, $R_1 = R_2 = 40\%$,  $\delta_{1} \in \{0\%$ (left), $100\%$ (right)$\}$, $\delta_{2} = 100\%$, $c_{1} = c_{2} = 1\%$, $\bar{\epsilon} = 10^{-5}$, $\Delta S = 0.05$, $\Delta t = 1/250$.}}
    \label{BCVA_Digital_Collateral_Funding}
\end{figure}

We illustrate  the bid-ask prices $P^b$ and $P^s$ of a call spread in Figure \ref{Bid-Ask_Stock_Digital}.} On  the left panel where the participant is assumed to be  default-free, we observe the  dominance of the three prices: $P^s \ge \Pi \ge P^b$.  However, in the bilateral counterparty-risk case, the  ordering of prices is different   in in-the-money (ITM) and out-of-money (OTM) ranges. We see that $\Pi \ge P^s \ge P^b$ in the ITM range, but $P^s \ge P^b \ge \Pi$  in  the OTM range.

\begin{figure}[th]
    \centering
    \includegraphics[width=0.45\textwidth]{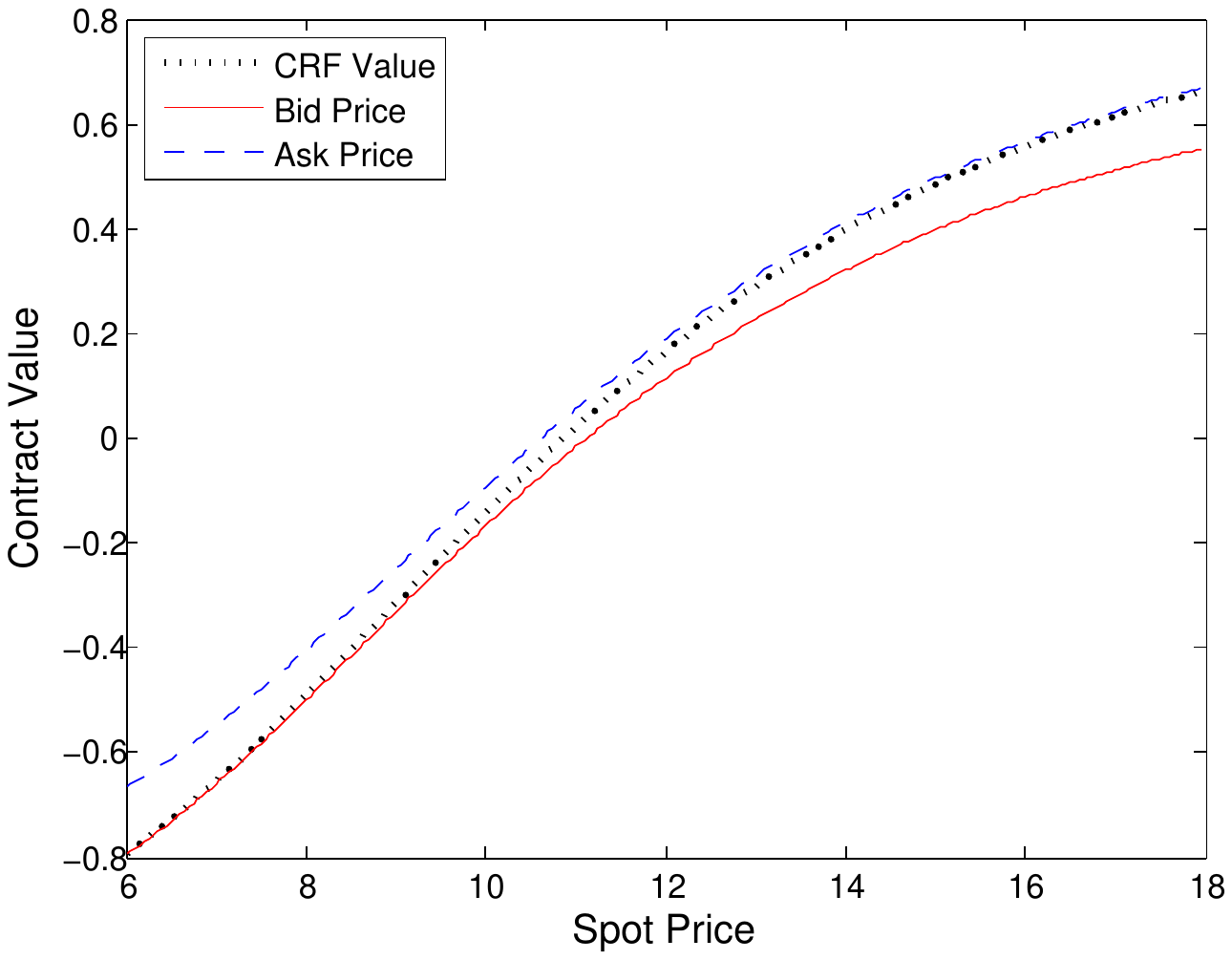}
    \includegraphics[width=0.45\textwidth]{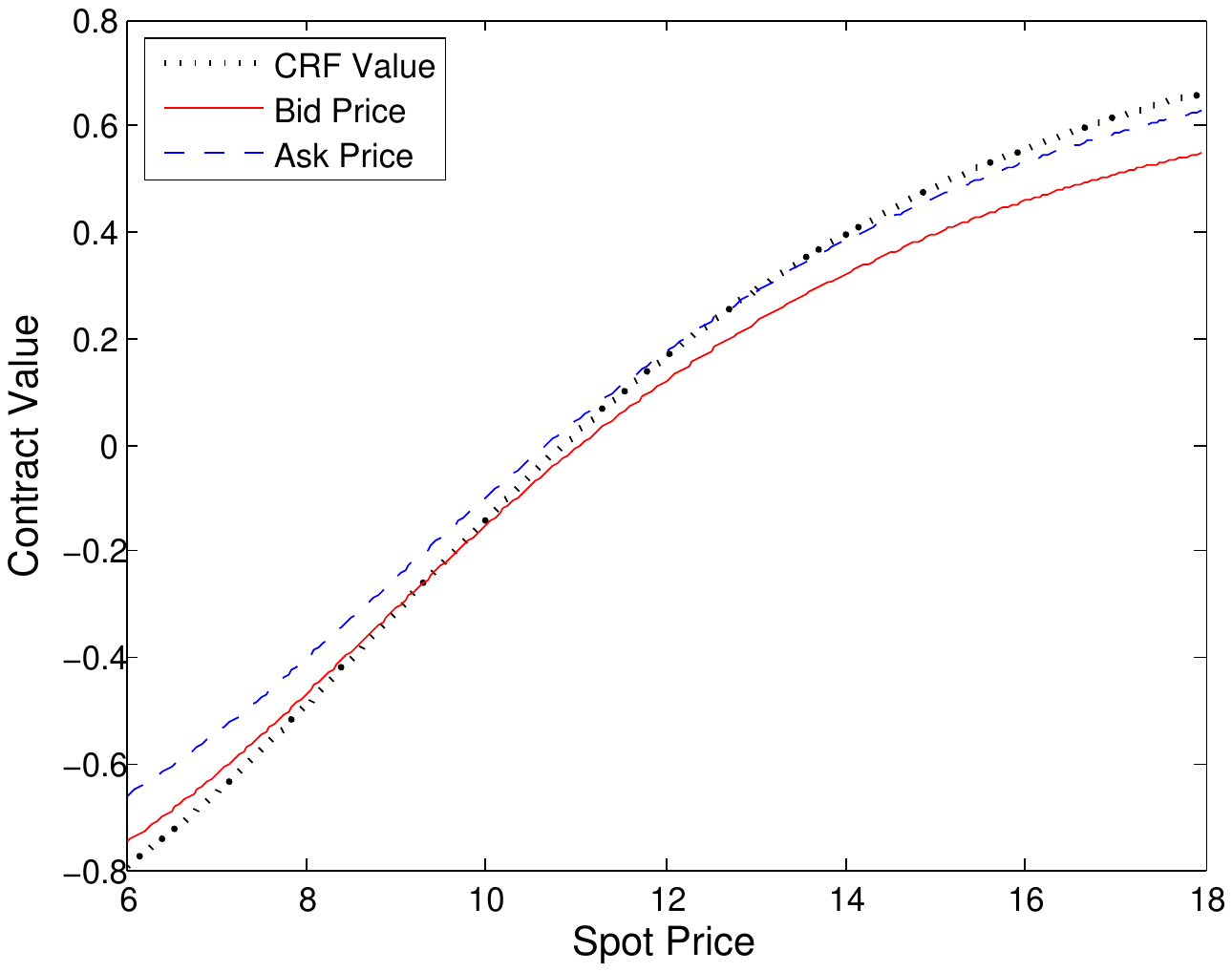}
	 \caption{\small{Call spread bid and ask prices with counterparty risk provision under (left) unilateral counterparty risk $\lambda^{(1)} = 0\%$, $\lambda^{(2)} = 15\%$, and (right) bilateral counterparty risk $\lambda^{(1)} = 5\%$\,, $\lambda^{(2)} = 15\%$. Other parameters are the same as in Figure \ref{BCVA_Digital_CounterpartyRisk_Default_Rate}.}}
    \label{Bid-Ask_Stock_Digital}
\end{figure}

\subsection{Equity Forwards}
Equity  forward contracts are commonly traded in  the OTC  market.  With stock $S$ as the underlying asset, we consider a  forward  with maturity $T$. The initial forward price  $F_0$ is set so that the contract has zero  value at inception. When the underlying stock defaults, the stock price goes to zero, and the buyer has to pay the discounted value  $e^{-r \, (T-\tau_0)} F_0$ at the default time. The contract cash flow is described by the triplet $(g,h,l)= (S_T-F_0,0,-e^{-r \, (T-\tau_0)}F_0)$. As the underlying stock price fluctuates over time, the MtM value also varies.  

The MtM value of a long forward contract $P(t,s)$ with provision (see \eqref{Expectation_With_Digital}) is computed using Algorithm \ref{Algorithm:Fixed_Point}. In Table \ref{table:convergence}, we show the convergence result of the MtM value at time $t =1$ when the stock price $S_1 = 20$, with initial forward price $F_0 = 10$. The first column of each case shows the value of the forward contract for each step $n \in \{1,...,6\}$. The second column of each case shows the error in terms of the supremum norm $\| P^{(n)} - P^{(n-1)}\|_{\infty}$ over  the whole domain $[0,T] \times \D$, and with tolerance $\bar{\epsilon}= 10^{-5}$ the algorithm stops at $n = 6$ in  both cases.  The number of iterations may depend on the initial value $P^{(0)}$, threshold $\bar{\epsilon}$ and upper bound $\bar{S}$. As we observe, for sufficiently large upper bounds $\bar{S} \in \{30, 40\}$, the convergent prices are the same.

\begin{table}[ht]
\begin{small}

  \centering
    \begin{tabular}{c||cc|cc}
	\hline
     $S_1 = 20$        & \multicolumn{2}{c|}{$\bar{S} = 30$}  & \multicolumn{2}{c}{$\bar{S} = 40$}\\
    \hline
          & Value & $\epsilon$  & Value & $\epsilon$ \\
    $n=0$   & 0     & -     & 0    & - \\
    $n=1$   & 8.6900 & 17.5592 & 8.6900 & 26.4284 \\
    $n=2$   & 7.6124 & 2.1391 & 7.6124 &    3.2034 \\
    $n=3$   & 7.6777 & 0.1289 & 7.6777 & 0.1927\\
    $n=4$   & 7.6751 & 0.0052 & 7.6751 & 0.0077\\
    $n=5$   & 7.6752 & 0.0002 & 7.6752 &  0.0002 \\
    $n=6$   & 7.6752 & $< 10^{-5}$ & 7.6752 & $< 10^{-5}$ \\
    \hline
    \end{tabular}%
      \caption{\small{Convergence of the values of a forward contract when spot price $S_1 = 20$ at $t =1$, with maximum stock price  $\bar{S} \in\{ 30$ (left column)$, $40 (right column)$ \}$. Other common parameters: $F_0 = 10$, $T = 3$, $r = 2\%$, $\sigma = 25\%$, $\lambda^{(0)} = 3\%$, $\lambda^{(1)} = 5\%$, $\lambda^{(2)} = 15\%$,  $R_1 = 40\%$, $R_2 = 40\%$, $\delta_{1} = \delta_{2} = 0$, $\bar{\epsilon} = 10^{-5}$, $\Delta S = 0.05$, $\Delta t = 1/500$.}}
\label{table:convergence} \end{small}
\end{table}

At time $t$, when the stock price is $s$, the CRF value of a long forward is given by
\begin{align}
\Pi(t,s) = (s - e^{- r \, (T-t)} F_0)\,. \label{Pi_stock_forward}
\end{align} 
In order to compute the MTM value of a long forward contract $\hP^b$ without counterparty risk provision, we apply \eqref{Pi_stock_forward} to  \eqref{Expectation_Without} and obtain 
\begin{align*}
\hP^b(t,s) &= \Pi(t,s) + \E_{t,s} \left[ \int_t^T e^{-(r+\lambda) (u-t)} \left( (\beta - \alpha) \, \Pi^+(u,S_u)\right) - \beta \, \Pi(u,S_u)  \, du \right] \notag \\
		   &= \Pi(t,s) + \E_{t,s} \left[ \int_t^T e^{-(r+\lambda) (u-t)} \left(  (\beta - \alpha) (S_u - e^{-r (T-u)} F_0)^+ - \beta (S_u - e^{-r (T-u)} F_0) \right)  du \right] \notag \\	
		   &= \Pi(t,s) + \int_t^T e^{-\lambda \, (u-t)} \left( (\beta - \alpha) \E_{t,s} \left[ e^{-r \, (u-t)} (S_u - e^{-r (T-u)} F_0)^+ \right] - \beta (s - e^{-r (T-u)} F_0)  \right) \,du. 
\end{align*}
To simplify the above equation, we notice that
$$\E_{t,s} \left[ e^{-r \, (u-t)} (S_u - e^{-r (T-u)} F_0)^+ \right] = C^{BS} (t, s\, ; T, e^{-r (T-u)} F_0,  r + \lambda^{(0)} , \sigma)\,.$$
We apply similar arguments to the seller's MtM value, and summarize as follows. 
\begin{proposition} The bid-ask prices without counterparty risk provision, $\hP^b(t,s)$ and $\hP^s(t,s)$,  of a stock forward  without counterparty risk provision  are given by
\begin{align*}
&\hP^b(t,s)  = \Pi(t,s) + \int_t^T\! e^{-  \lambda \, (u-t)} \left( (\beta - \alpha) \,  C^{BS} (u, s \,; T, e^{-r (T-u)} F_0,
 r + \lambda^{(0)} , \sigma)  - \beta (s - e^{-r (T-t)} F_0 ) \right)  du,  \\
&\hP^s(t,s)  = \Pi(t,s) + \int_t^T\! e^{-  \lambda \, (u-t)} \left( (\alpha - \beta) \,  C^{BS} (u, s \,; T,  e^{-r (T-u)} F_0,
 r + \lambda^{(0)} , \sigma)  - \alpha (s - e^{-r (T-t)} F_0 ) \right)  du,
\end{align*}
with  $\Pi(t,s)$ satisfying \eqref{Pi_stock_forward}.
\end{proposition}

The fair forward price makes the MtM value of the contract equal to zero at the start of the contract. The  CRF value in  \eqref{Pi_stock_forward}  implies that  $F_0 = e^{r \, T} S_0$. However, the fair forward price in presence of bilateral counterparty risk is found implicitly.  Precisely, the fair forward price $ F^*_0$ of a long (resp. short) position makes the MtM value with counterparty risk provision satisfies  $P^b(0,S_0; F^b_0) = 0$ (resp.  $P^s(0,S_0; F^s_0)$ = 0).

In Figure \ref{Bid-Ask_Stock_Forward_GBM}, we plot the bid-ask prices $P^b$ and $P^s$ of the forward contract with counterparty risk provision together with the CRF value. On the left panel, all three values increase as the underlying stock price increases. Similar to the call spread case,  the price ordering   changes from  $\Pi \ge P^s \ge P^b$ in  the ITM range to $P^s  \ge P^b \ge \Pi$ in the OTM range. On the right panel, both MtM values decrease significantly as the counterparty default rate increases. However, the MtM with provision moves more rapidly, similar to Figure \ref{BCVA_Digital_CounterpartyRisk_Default_Rate}.

\begin{figure}[htp]
    \centering
    \includegraphics[width=0.45\textwidth]{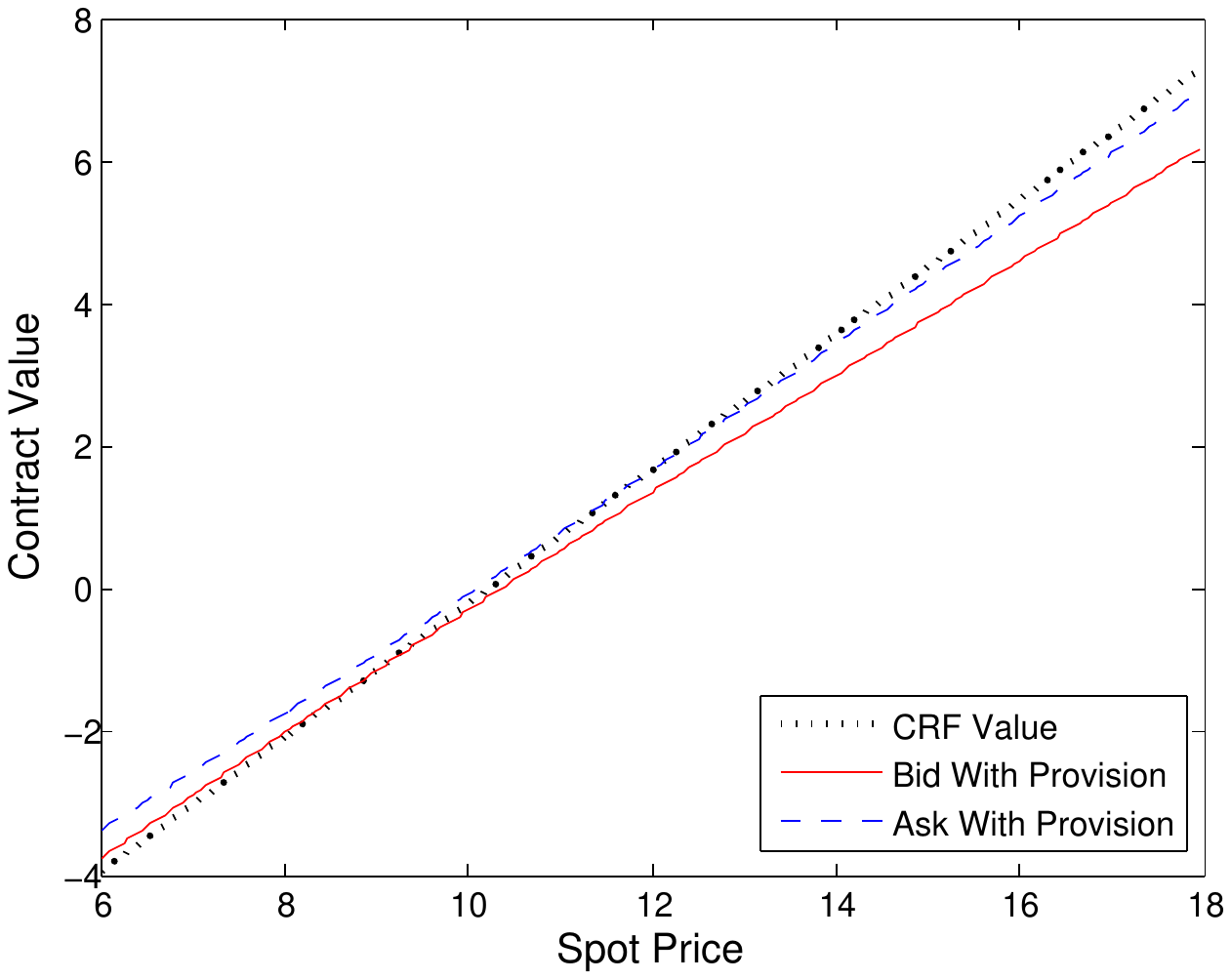}
    \includegraphics[width=0.45\textwidth]{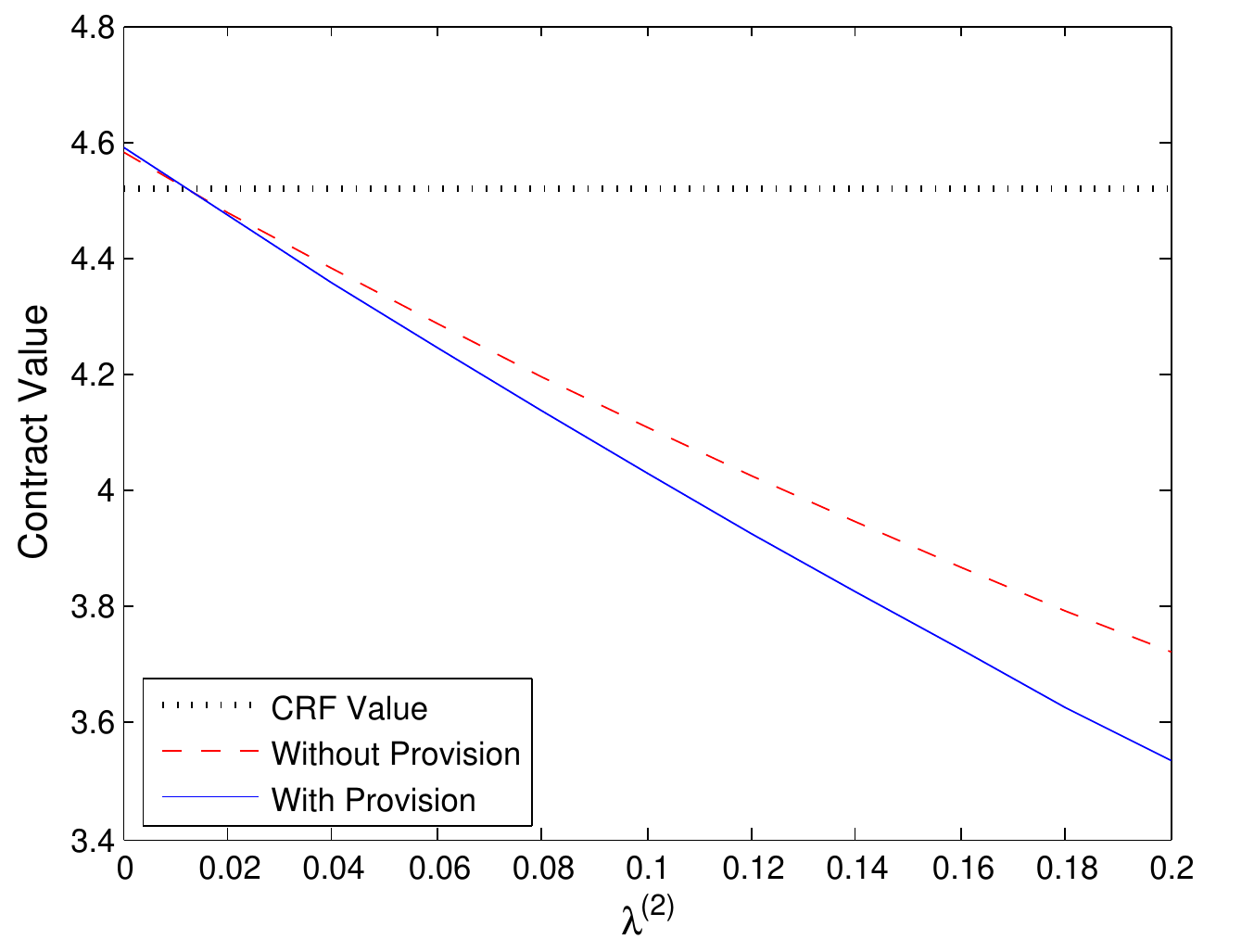}
	 \caption{\small{(Left) The value of a stock forward at time $0$ is increasing in the spot price. The fair forward price is the spot price  at which the contract value is zero. (Right) The MtM forward contract values of a stock forward  are decreasing in the counterparty's default rate $\lambda^{(2)}$, with $F_0 = 10$, $t=1$ and $T=3$.  Parameters: $r = 2\%, \sigma = 25\%$, $\lambda^{(0)} = 3\%$,  $\lambda^{(1)} = 5\%$, $R_1 = 40\%$, $\lambda^{(2)} = 15\%$, $R_2 = 40\%$,  $\delta_1 = \delta_2 = 0$, $\bar{\epsilon} = 10^{-5}$, $\Delta S = 0.05$, $\Delta t = 1/500$.}}
    \label{Bid-Ask_Stock_Forward_GBM}
\end{figure}

\begin{remark}[Total Return Swap]
Total return swaps (TRS) are also traded over the counter as an alternative to  equity stock forwards. The swap buyer will receive the increase in equity value at swap expiration date $T$, $S_T - S_0$. On the other hand, the buyer continuously pays a premium rate $p \ge r$ until the expiration date and also pays the decrease in the equity value at the expiration date to the swap seller. The TRS is represented as the triplet $(S_T-S_0, - p, -S_0 \, e^{-r \, (T - \tau_0)})$.
\end{remark}

\subsection{Claims with Positive Payoffs}
Some derivatives have positive payoffs, i.e. the triplet $g,h,l \ge 0$, including calls, puts and digital options. For both conventions, the nonlinear property of the price function disappears since we can substitute the nonlinear functions  $P^+(t,s)$ and $\Pi^+(t,s)$ in \eqref{PDE_With} and \eqref{PDE_Without}  by the linear functions $P(t,s)$ and $\Pi(t,s)$. From this observation, we derive the formulae for the bid-ask prices of derivatives with positive payoffs.
\begin{proposition} \label{Prop:UCVA_Positive_NPV} 
For any claim with $g,h,l \ge 0$, the bid-ask prices of with counterparty risk provision satisfy
\begin{align}
	 P^b(t,s) &= \E_{t,s} \bigg[ e^{- (r + \alpha + \lambda^{(0)}) \, (T-t)} \, g(S_T) +  \int_t^T  e^{-  (r + \alpha + \lambda^{(0)}) \, (u-t)} (h(S_u) + \lambda^{(0)} \, l(u,S_u)) \, du \bigg], \label{with_option_buyer}\\
	P^s(t,s) &= \E_{t,s} \bigg[ e^{-  (r + \beta + \lambda^{(0)}) \, (T-t)} \, g(S_T) + \int_t^T  e^{- (r + \beta + \lambda^{(0)}) \, (u-t)} (h(S_u) + \lambda^{(0)} \, l(u,S_u)) \, du \bigg].  \label{with_option_seller}
\end{align}
Without counterparty risk provision (see Remark \ref{remark:without}), the bid-ask prices satisfy
\begin{align}
 \hP^b(t,s) &= \Pi(t,s) - \E_{t,s}\left[ \int_t^T \alpha  \,  e^{-  (r + \lambda) \, (u-t)} \,  \Pi(u,S_u) \, du \right], \label{without_option_buyer}\\
 \hP^s(t,s) &=  \Pi(t,s) - \E_{t,s}\left[ \int_t^T \beta  \,  e^{-  (r + \lambda) \, (u-t)} \,  \Pi(u,S_u) \, du \right]. \label{without_option_seller}
\end{align}
\end{proposition}
We apply these results to call options on a defaultable stock.

\begin{example} The bid price of a European call option with provision is given by
 \begin{align}
 P^{b}(t,s) &= e^{ -\alpha (T-t)} \, C^{BS}(t, s \,; T, K,r + \lambda^{(0)},\sigma)\,. \label{call_with}
\end{align} 
Moreover, the bid price of a European call option without provision is given by
\begin{align}
 \hP^{b}(t,s) &= \left(1 - \frac{\alpha}{\lambda^{(1)} + \lambda^{(2)}} \left( 1- e^{- (\lambda^{(1)} + \lambda^{(2)}) \, (T-t)}\right) \right) \, C^{BS}(t, s \,;T , K,r + \lambda^{(0)},\sigma)\,. \label{call_without}
\end{align}
\end{example}
 In both \eqref{call_with} and \eqref{call_without}, the bid prices have two components: the CRF value of the call option, i.e. Black-Scholes price, and a multiplier that depends on the parameter $\alpha$ (see \eqref{alpha_constant}). First, suppose that $\alpha = 0$. This happens, for example, when the contract is perfectly collateralized and effective collateral rate is zero, i.e. $\delta_2 = c_2 = 0$. In this case, both bid prices are identical to the CRF value $\Pi(t,s) = C^{BS}(t,s\,;T,K,r+\lambda^{(0)},\sigma)$. The parameter $\alpha$ is positive  when the contract is under-collateralized by party $2$, i.e. $\delta_2 < 1$. In this case, the call option buyer is exposed to counterparty default risk, and consequently, the bid prices are less than the CRF value $\Pi(t,s)$. On the other hand, $\alpha$ becomes negative when the contract is over-collateralized by party $2$ and the collateral rate is negligible, i.e. $\delta_2 > 1$ and $c_2 = 0$. The call buyer receives additional financial benefit from excess collateral since the buyer only returns a fraction $1 - L_1$ of the over-collateralized amount at the buyer's own default. As a result, the bid prices are greater than the CRF value. 

\begin{figure}[htp]
    \centering
        \includegraphics[width=0.45\textwidth]{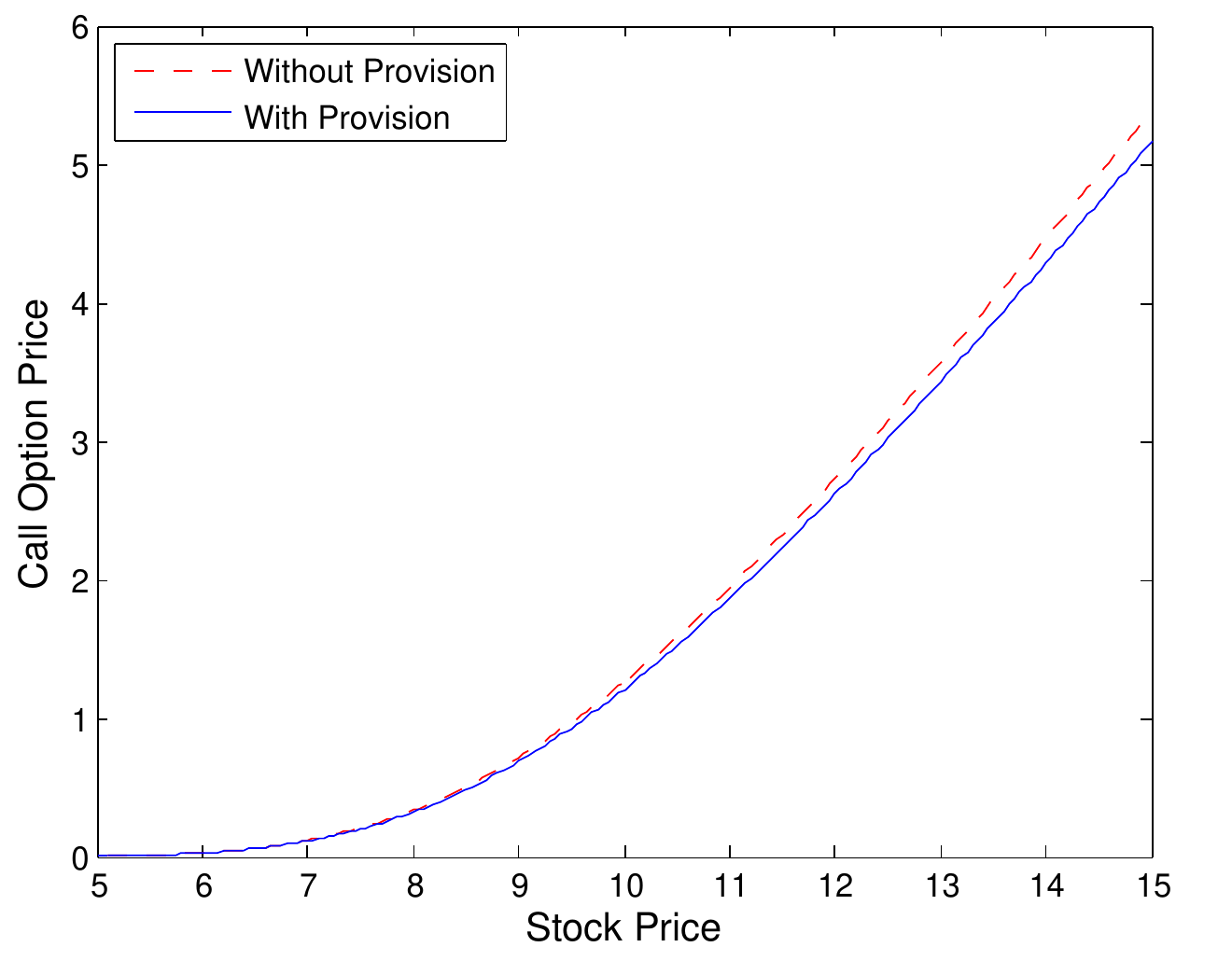}
    \includegraphics[width=0.45\textwidth]{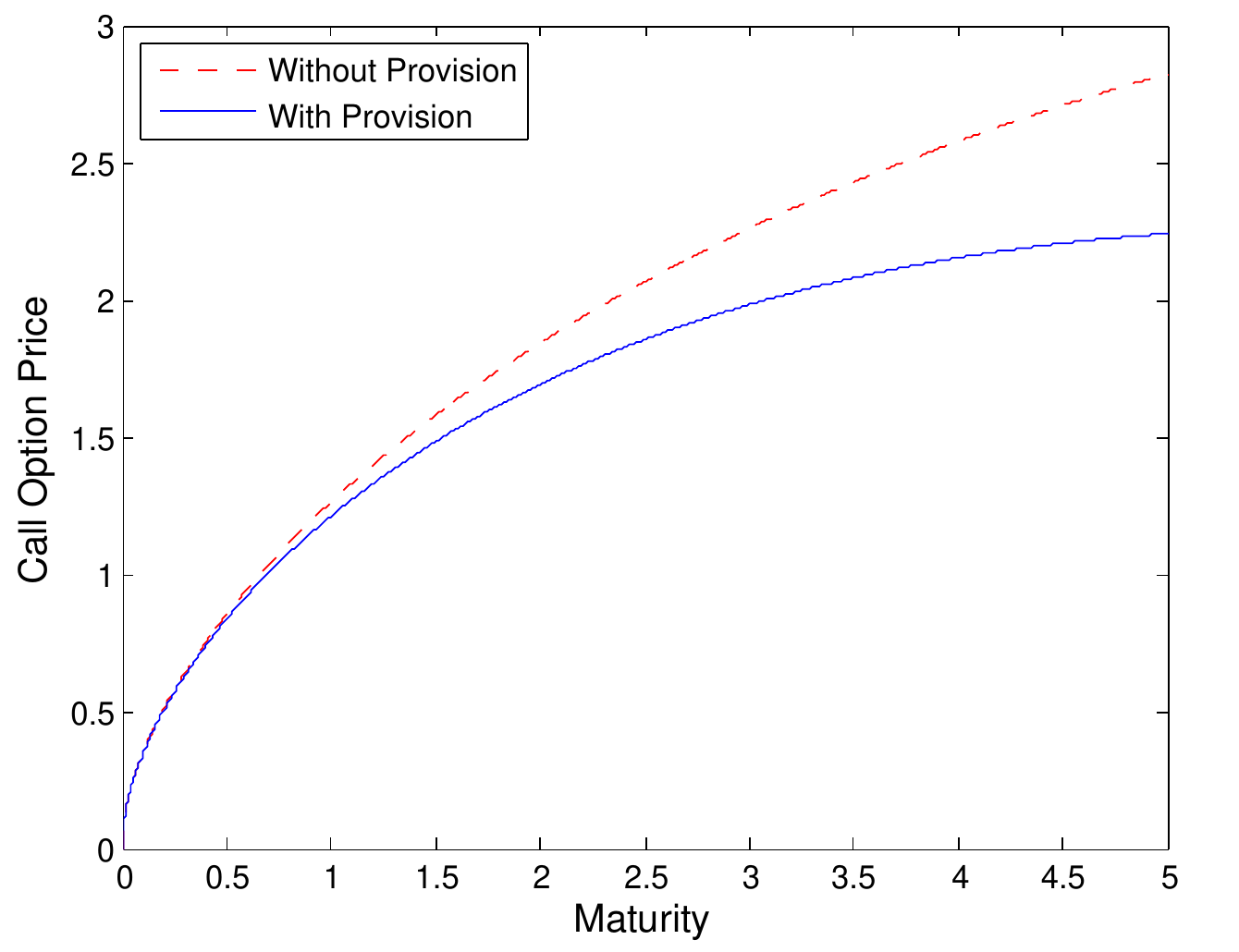}
  	 \caption{\small{The bid prices of a call option with and without counterparty risk provision are increasing in stock price $s \in [5,15]$ (left) and in  maturity $T \in [0,5]$ (right).  Parameters: $\lambda^{(0)} = 5\%,\, \lambda^{(2)} = 10\%, \, S_0 = 10, \, T = 1, \, t=0,\,  r = 2\%, \, \sigma = 25\%, \, R_1 = R_2 = 40\%$, $K = S_0 = 10$\, and $\delta_1 = \delta_2 = 0.$}}
    \label{CVA_Call_Jump2}
\end{figure}

Figure \ref{CVA_Call_Jump2} illustrates the two bid prices (with and without provision) for various maturities and spot prices. We observe the dominance of the price without provision over the price with provision. Moreover, the difference of two prices increase as maturity or spot price increases. This wider level of difference is attributed to the fact that the difference of counterparty risk exposure increases as the maturity increases. 

Next, we derive a number of price dominance relationships for the bid-ask prices and the CRF value of claims with positive payoffs.
\begin{proposition} \label{prop:bid_ask_dominance} Consider any claim with the triplet $g, h, l\ge 0$. If $\alpha , \beta \ge 0$, then  the following price dominance relationships hold:
\begin{align}
\hP^b(t,s), \hP^s(t,s) \le \Pi(t,s) \quad \text{and} \quad  P^b(t,s) , P^s(t,s) \le  \Pi(t,s)\,. \label{bid_ask_relationship}
\end{align}
Furthermore,  if $\lambda^{(1)} + \lambda^{(2)} \ge \alpha, \beta  \ge 0$, then we have 
\begin{align}
 P^b(t,s) \le \hP^b(t,s) \le \Pi(t,s) \quad \text{and} \quad P^s(t,s) \le \hP^s(t,s) \le \Pi(t,s)\,. \label{bid_ask_relationship2}
\end{align}
\end{proposition}
We provide a proof in Appendix \ref{sect-prof46}. In  market conditions  where the effective collateral rates $c_1, c_2$ are much lower than the  counterparty/own default rates $\lambda^{(1)}, \lambda^{(2)}$, the conditions for \eqref{bid_ask_relationship} and \eqref{bid_ask_relationship2} are satisfied. In turn, since the buyer's (resp. seller's) MtM value  with provision assumes  $P^b$ (resp. $P^s$) as the liquidation value, which is lower than   the liquidation value $\Pi$ in the case without provision according to  \eqref{bid_ask_relationship}, this implies  the price dominance relationships in  \eqref{bid_ask_relationship2}, as is presented in Figure \ref{CVA_Call_Jump2}.

\section{Defaultable Fixed-Income Derivatives with Counterparty Risk} \label{sec:BCVA_Credit}
We now consider defaultable fixed-income claims under bilateral counterparty risk setting. We model the stochastic factor $X$ by two affine diffusions: Ornstein-Uhlenbeck (OU) and Cox-Ingersoll-Ross (CIR) processes. Under the OU model, the factor process $X$ follows
\begin{align*}
  dX_t &= \kappa ( \theta - X_t) \, dt + \sigma  \, d\tW_t\,, 
\end{align*}
with positive constant parameters $\kappa, \,\theta, \,\sigma > 0$ that represent the speed of mean reversion, long-term mean, and volatility, respectively. In the CIR model, the factor process $X$ follows
\begin{align*}
  dX_t &= \kappa ( \theta - X_t) \, dt + \sigma \sqrt{X_t} \, d\tW_t\,, 
\end{align*}
where the positive parameters $\{ \kappa, \theta, \sigma\}$ satisfy $\kappa \,, \theta\,, \sigma > 0$ and the Feller condition $2 \kappa \, \theta > \sigma^2$ so that $X$ stays positive at all times a.s.  Both models for $X$ satisfy the corresponding conditions among (C1)-(C7) \citep[Sections 2.1-2.2]{heath2000martingales}, and we will consider here swaps whose payoff are bounded continuous.

A generic fixed-income contract is  described by the triplet $(g(x),h(x),l(t,x))$, with  default times $\{\tau_i\}_{i=0,1,2}$ as in  \eqref{def:default_time_dist} and Markovian  intensities of the form $\lambda^{(i)}_t = \lambda^{(i)}(t,X_t)$. From \eqref{Expectation_With}, the MtM value \textit{with counterparty risk provision} is given by
\begin{align}
	P(t,x) &= \E_{t,x} \bigg[ e^{- \int_t^T (r(u) + \lambda(u,X_u)) \, du} \, g(X_T) +  \int_t^T  e^{- \int_t^u (r(v) + \lambda(v,X_v)) \, dv} f(u,X_u,P_u) \, du 
					  \bigg], \label{Expectation_With_Fixed_Income}
\end{align}
where $f(t,x,y) = h(x) + \lambda^{(0)}(t,x) \, l(t,x) + (\lambda^{(1)}(t,x) + \lambda^{(2)}(t,x) - \beta(t,x)) y + (\beta(t,x) - \alpha(t,x)) \, y^+$. The parameters $\alpha(t,x)$ and $\beta(t,x)$ are defined similarly to the definitions \eqref{alpha} and \eqref{beta}. Noticing that $P(t,x)$ in  \eqref{Expectation_With_Fixed_Income} is a special case of $P(t,s,x)$ in \eqref{Expectation_With}, we use Algorithm \ref{Algorithm:Fixed_Point} to find the MtM value by solving iteratively the following PDE:
\begin{align}
	&\frac{\partial P^{(n)}}{\partial t} + \L_X P^{(n)} - (r(t) + \lambda(t,x)) \, P^{(n)}  + f(t,x,P^{(n-1)}) = 0\,, \label{PDE_With_Credit_Iteration}
\end{align}
for $(t,x) \in [0,T) \times \R$, with $P^{(n)}(T,x) = g(x)$ for $x \in \R$ ($x \in \R_+$ for the CIR case). 
The MtM value without counterparty risk provision $\hP(t,x)$ is similarly obtained by replacing $P(u,X_u)$ with $\Pi(u,X_u)$ in the right hand side of \eqref{Expectation_With_Fixed_Income}.


 For numerical examples, we assume the constant interest rate $r$ and constant default rates $\lambda^{(1)}$ and $\lambda^{(2)}$.  The reference default intensity is modeled by   $\lambda^{(0)}_t= (X_t)^+ \wedge\, \overline{X}$ under the OU or CIR model for $X$. In other words, we  impose an   upper  bound $\overline{X}$ and a lower bound  $\underline{X}=0$ on the reference default intensity, and this allows us to  apply our contraction mapping result in Theorem \ref{prop:PDE_classical_solution}.  For implementation, we apply the FDM and Algorithm \ref{Algorithm:Fixed_Point} in Section \ref{numerical_solution} to \eqref{PDE_With_Credit_Iteration} to obtain the MtM value with counterparty risk provision.
 

\subsection{Credit Default Swaps (CDS)}\label{sect-cds}
After the 2008 financial crisis, CDS contracts are traded in a new standardized way whereby the protection buyer pays at a fixed premium rate, along with    a non-zero (positive/negative) upfront payment at the start of the contract.
A long position of a CDS contract written on the default event $\{\tau_0 \le T\}$ is summarized by the triplet notation $(0,-p,1)$. 


  In Table \ref{table:covergence_CDS}, we show the convergence of the CDS MtM values. The first and  second columns for  each model show the values of the CDS at $x = 2\%$ and $x= 8\%$, respectively. In addition, the third column shows the error in terms of the supremum norm $\epsilon = \| P^{(n)} - P^{(n-1)}\|_{\infty}$ over the entire grid for each step $n= 1, 2, \ldots,  7$. The algorithm stops at $n = 5$ under the OU model, and at $n =7$ under the CIR model with the common tolerance level $\bar{\epsilon} = 10^{-5}$. 

\begin{table}[ht]
\begin{small}
  \centering
   \begin{tabular}{c|ccc|ccc}
    \hline
         & \multicolumn{3}{c|}{OU}      & \multicolumn{3}{c}{CIR} \\
    \hline
    \hline
          & $P^{(n)}(0,2\%)$ & $P^{(n)}(0,10\%)$ & $\epsilon$  & $P^{(n)}(0,2\%)$ & $P^{(n)} (0,10\%)$ & $\epsilon$ \\
    $n=0$   & 0    & 0 & -     & 0    & - \\
    $n=1$   & -0.1326 & 0.05218 & 0.9048 & -0.1130 & 0.1569 & 0.3950 \\
    $n=2$   & -0.1526 & 0.04518 & 0.0992 &  -0.1131 & 0.1176 & 0.0911 \\
    $n=3$   & -0.1515 & 0.04562 & 0.0060 & -0.1130 & 0.1248 & 0.0159 \\
    $n=4$   & -0.1516 & 0.04560& 0.0002 & -0.1130 & 0.1238 & 0.0022 \\
    $n=5$   & -0.1516 & 0.04560& $< 10^{-5}$ & -0.1130 & 0.1239 & 0.0003\\
    $n=6$   & - & -  & - &-0.1130 & 0.1239 & $< 10^{-4}$\\
    $n=7$   & - & -  & - & -0.1130 & 0.1239 & $< 10^{-5}$ \\
    \hline
    \end{tabular}%
      \caption{\small{Covergence of MtM values of a CDS in the OU/CIR model with reference default rate $x = 2\%$ and $x = 10\%$. For the OU model: $\theta = 3\%$, $\sigma = 3.5\%$, $\kappa = 5\%$. For the CIR model: $\theta = 3\%$, $\sigma = 5\%$, $\kappa = 5\%$. Other common Parameters: $T = 5$, $t= 0$, $p = 100bps$, $r = 2\%$,  $\psi_{0} = 0$, $\psi_{1} = 5\%$, $\psi_{2} = 25\%$, $w_0 = 1$, $w_1 = w_2 = 0$, $\lambda^{(1)} = 5\%$, $\lambda^{(2)} = 25\%$, $R_1 = R_2 = 40\%$, $\delta_1 = \delta_2 = 0$, $\bar{\epsilon} = 10^{-5}$, $\overline{X} = 20\%$, $\underline{X} = 0$, $\Delta X = 0.001$, $\Delta t = 1/500$.}}
\label{table:covergence_CDS} 
\end{small}
\end{table} 
\begin{figure}[h!]
    \centering
    \includegraphics[width=0.45\textwidth]{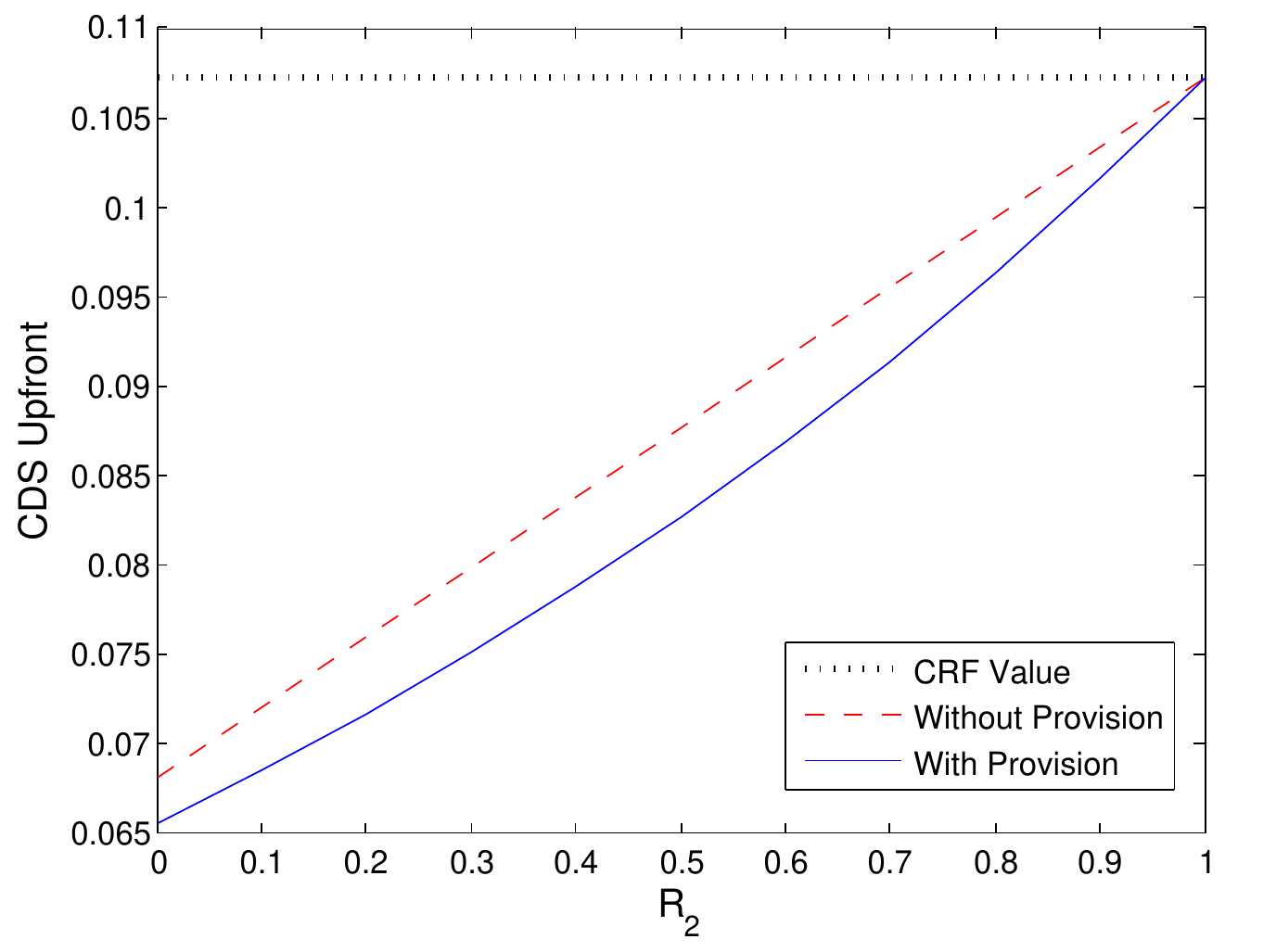}
    \includegraphics[width=0.45\textwidth]{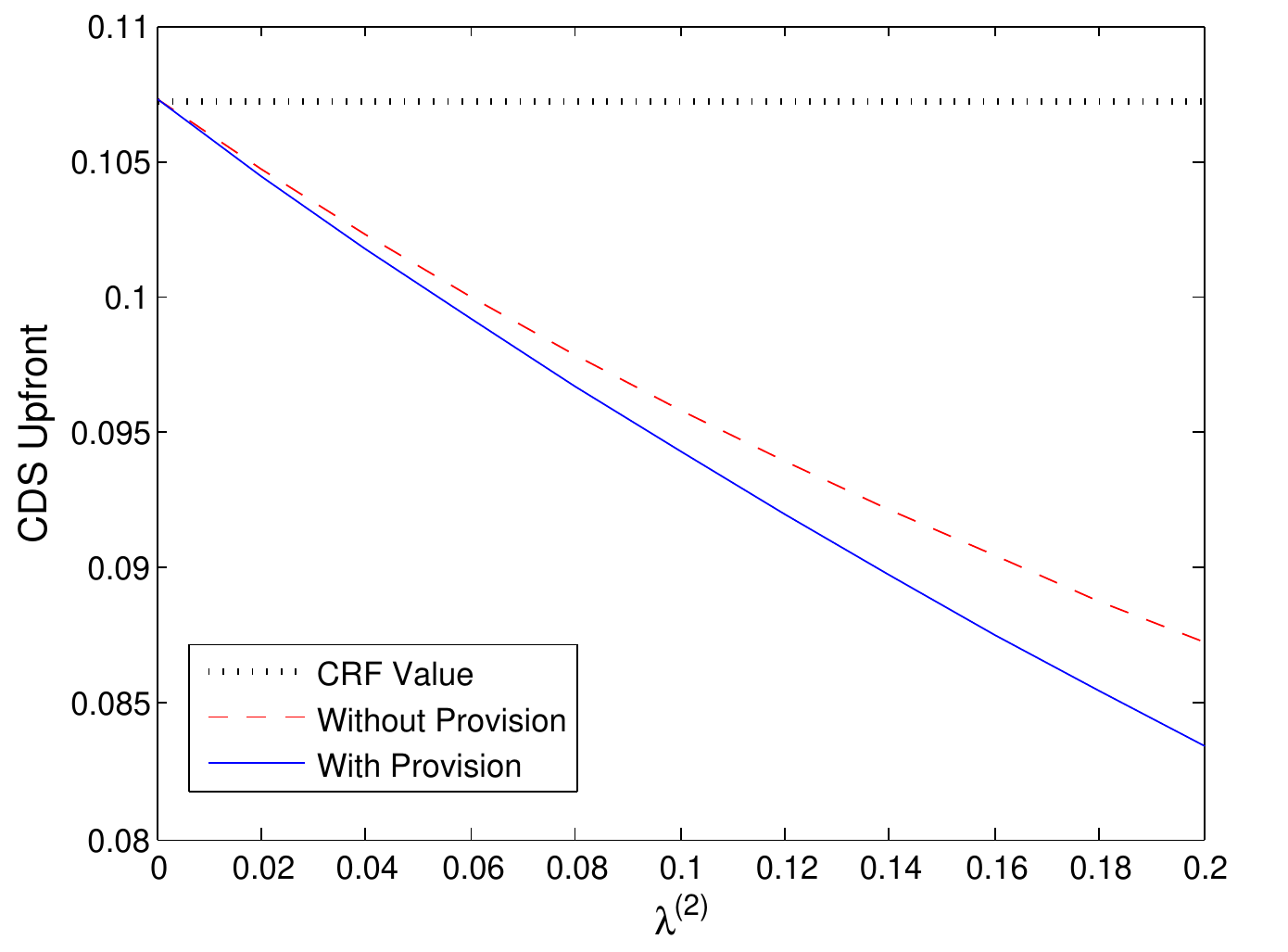}
	 \caption{\small{The CDS upfront prices under the CIR Model  are increasing in the counterparty recovery rate $R_2$ (left) and   decreasing in the counterparty default rate $\lambda^{(2)}$ (right). Other parameters are same as in Table \ref{table:covergence_CDS} along with $x = 8\%$.}}
    \label{BCVA_CDS_Convergence}
\end{figure}

In Figure \ref{BCVA_CDS_Convergence}, we compare the three MtM values in terms of the counterparty recovery rate $R_2 = 1 - L_2$ and the counterparty default rate $\lambda^{(2)}$. The CRF value is given by \eqref{eq:CDS} and is independent of $R_2$ and $\lambda^{(2)}$, so it is flat across the $x$-axes in  Figure \ref{BCVA_CDS_Convergence}. An increase in counterparty recovery rate or a decrease in counterparty default rate improves both MtM values, whereas the CRF value does not change. In both cases, the MtM value without provision dominates the MtM value with provision.

\subsection{Total Return Swaps (TRS)}\label{sect-trs}
Total return swaps (TRS) on defaultable bonds are OTC traded, and their MtM values are subject to counterparty risk.  A TRS is also referred to as  a bond forward \citep[see][Chap.\,2.5]{schonbucher2003credit}.  Fix a maturity of  $T$ years, we consider a TRS on a zero-recovery defaultable bond with  maturity  $T' \ge T$. The value of the defaultable bond, denoted by $C$,  is given in \eqref{eq:zero_coupon_bond_OU} for the OU model  and in \eqref{eq:zero_coupon_bond} for the CIR model. Given no default up to  the swap maturity $T$, the  swap buyer receives the difference between the bond price $C(T,X_{T}\,;\,T')$ and the pre-specified strike $K = C(0,X_0;T')$ from the swap seller. If the bond defaults before time $T$, the swap buyer pays the strike $K$ to the seller. Until  the first default time $\tau$  or expiration date $T$, the buyer continues to pay  at the risk-free rate plus a spread $r + p$.  For the MtM upfront value $P(t,x)$ with counterparty risk provision, we follow  \eqref{Expectation_With_Fixed_Income} with  the  triplet  \begin{align*}
 g(x) &= C(T,x\,;\,T') - K, \quad  h(t,x) = -K (r+ p), \quad
 l(t,x) = - K.
\end{align*} 

\begin{figure}[htp]
    \centering
    \includegraphics[width=0.45\textwidth]{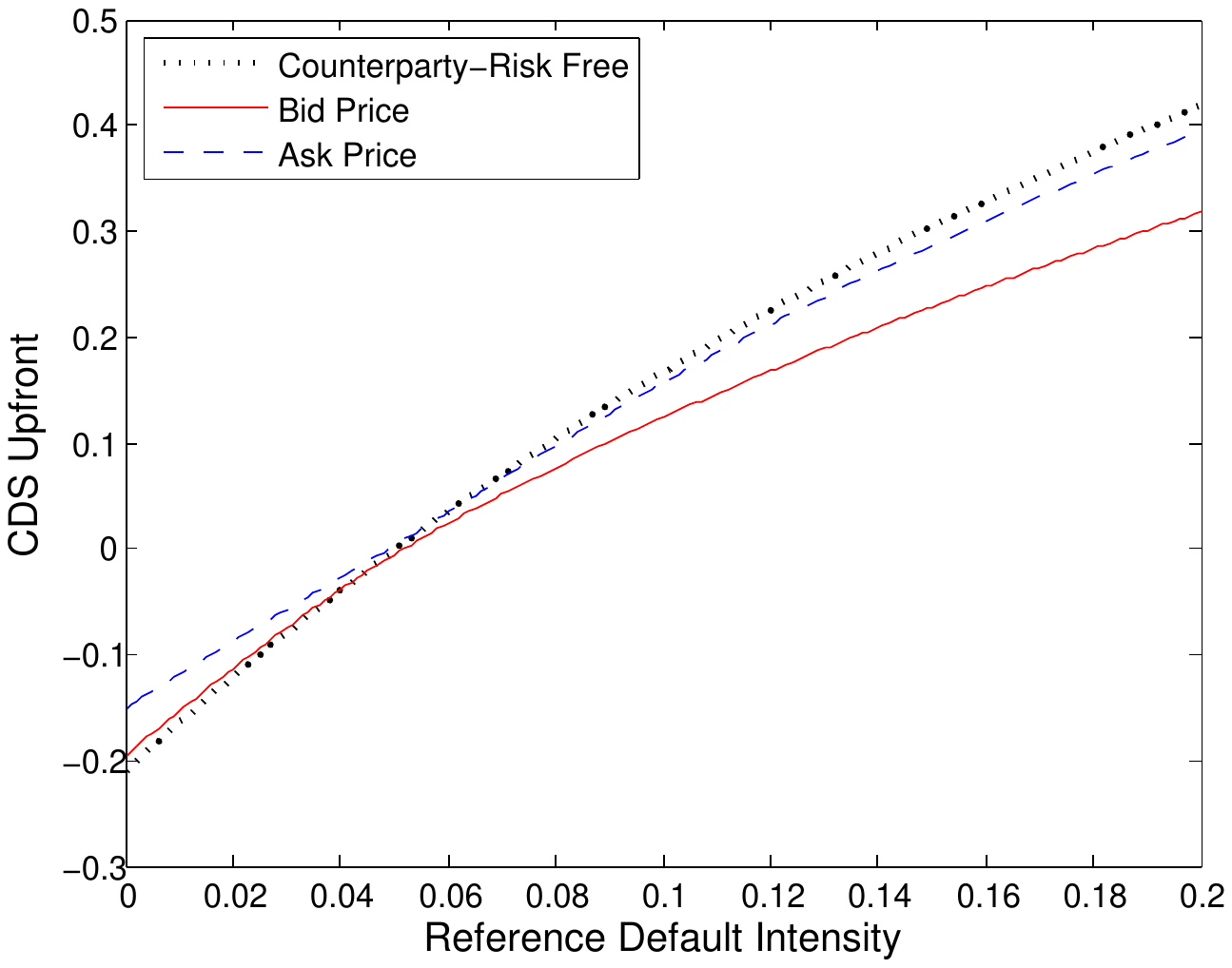}
    \includegraphics[width=0.45\textwidth]{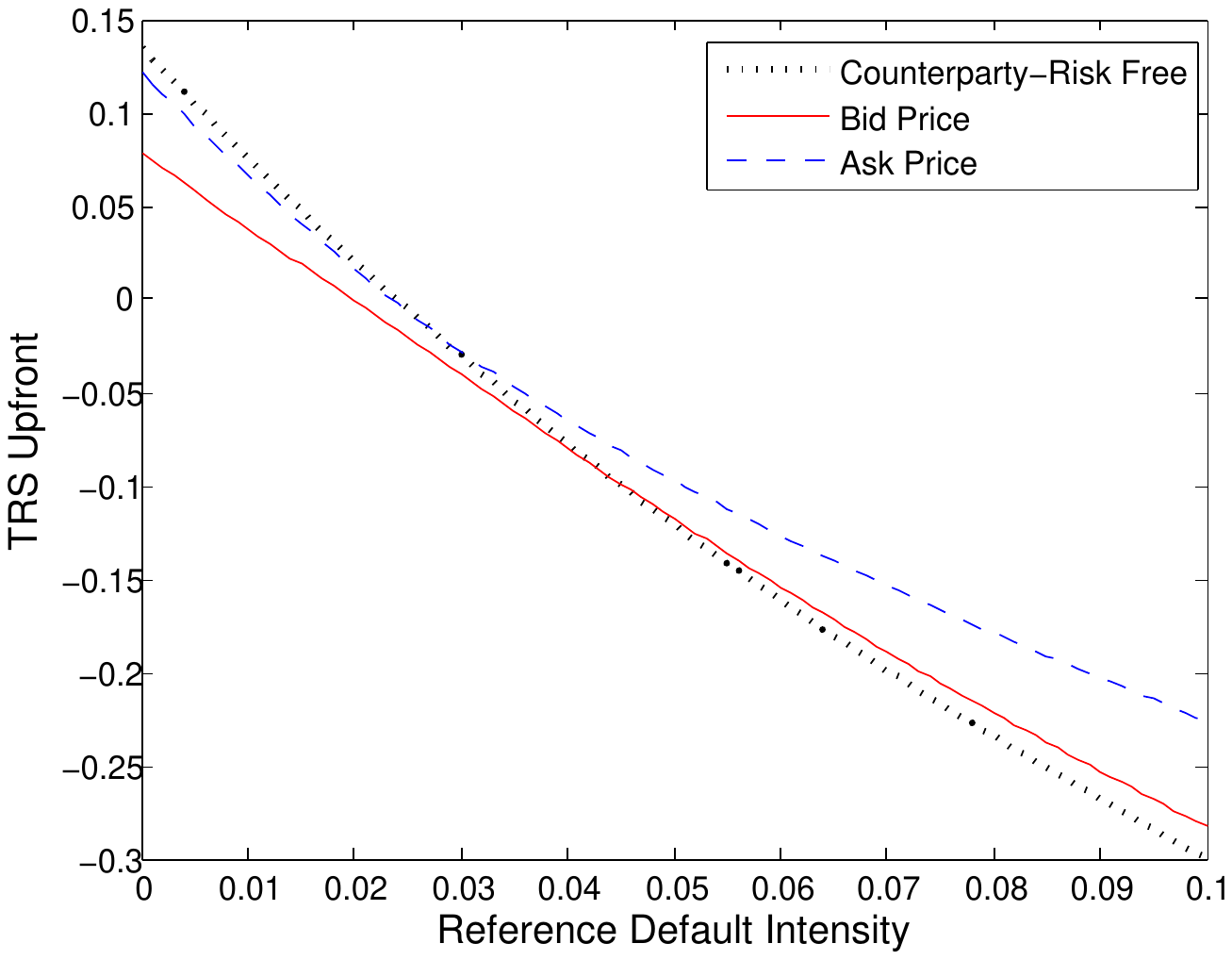}
	 \caption{\small{(Left) CDS bid-ask upfront prices under the CIR model. (Right) TRS bid-ask upfront prices under the CIR model. Other common parameters: $x = 8\%$, $T= 3$, $T' = 10$, $t = 0$, $p = 100bps$, $r = 2\%$, $\theta = 3\%$, $\kappa = 5\%$,  $\psi_{0} = 0$, $\psi_{1} = 5\%$, $\psi_{2} = 25\%$, $w_0 = 1$, $w_1 = w_2 = 0$, $\lambda^{(1)} = 5\%$, $R_1 = 40\%$, $\lambda^{(2)} = 25\%$, $R_2 = 40\%$, $\delta_1 = \delta_2 = 0$, $\bar{\epsilon} = 10^{-8}$,  $\overline{X} = 20\%$, $\underline{X} = 0$,  $\Delta X = 0.001$, $\Delta t = 1/500$.}}
    \label{BCVA_TRS_CIR_OU}
\end{figure}

In Figure \ref{BCVA_TRS_CIR_OU}, we obtain the  bid-ask upfront prices  with provision for a CDS and a TRS, along with the counterparty risk-free upfront values, under the CIR model. Since both CDS and TRS are swaps, the MtM values can be positive or  negative, depending on the reference default intensity.   As buying a CDS is similar to longing default risk, the CDS upfront value is increasing in the reference default intensity (left). On the other hand, since the TRS buyer is shorting default risk, the TRS upfront value decreases as the reference asset's default intensity $\lambda^{(0)}$ increases (right). 
\section{Conclusion} \label{Conclusions}
In summary, we have discussed a valuation framework to analytically study and numerically compute financial contracts subject to reference and counterparty
default risks with bilateral  collateralization. In addition to the underlying price dynamics used in this paper, our fixed point approach and the corresponding iterative numerical  algorithm can potentially be adapted to  price derivatives with counterparty risk in other models. The challenge lies in efficiently and accurately solving  a sequence of inhomogeneous PDE problems. Our model also sheds light on the role played by counterparty risk and collateralization in the formation of bid-ask spreads. 

Several interesting and practically important  problems remain for future investigation. For example, what are the price and risk impacts of  utilizing \emph{multiple} counterparties for OTC trading? A recent study by \cite{capponi13} derives the XVA of a CDS portfolio  with  a large number of entities. Apart from pricing,  counterparty risk should also be incorporated into  static or dynamic  strategies for derivatives trading (see e.g. \cite{jiaopham}), as well as mark-to-market timing \cite{leung2012risk}. In view of the financial crisis, counterparty risk has also become a key component in the design of clearing houses. Answers to these questions will be useful not only for individual or institutional investors, but also for regulators. 

\singlespacing
\appendix
\section{Appendix} \label{sec:Proofs}
\subsection{Proof of Theorem \ref{prop:PDE_classical_solution}} \label{sect-proofthm}
For $w \in C_b^{1,2}([0,T) \times \D,\R)$ and the operator $\M$ in \eqref{def:M}, we define $v \equiv v(t,s,x)$ by
\begin{align}
	v &:= \M w = \E_{t,s,x} \bigg[ e^{- \int_t^T \tilde{r}_u \, du} \, g(S_T,X_T) +  \int_t^T  e^{- \int_t^u \tilde{r}_v \, dv} f(u,S_u,X_u,w(u,S_u,X_u)) \, du  \bigg] \label{interim_proof_theorem}.
\end{align}
Equation \eqref{interim_proof_theorem} admits the same form as $(1.2)$ of \cite{heath2000martingales}. By Theorem $1$ of \cite{heath2000martingales}, $v$ is a classical solution $(C^{1,2}([0,T) \times \D, \R))$ of the PDE:
\begin{align}
 \frac{\partial v(t,s,x)}{\partial t} + \L v(t,s,x) - \tilde{r}(t,s,x) \, v(t,s,x)  + f(t,s,x,w(t,s,x)) = 0\,, \label{PDE_With_iteration_w_v}
\end{align}
for $(t,s,x) \in [0, T) \times D$ under certain conditions. To apply their result, we verify the sufficient conditions
$(A1)$,$(A2)$,$(A3')$ and
$(A3a')-(A3e')$ in their Theorem $1$. The conditions $(A1),(A2),(A3')$ and $(A3a')-(A3c')$ are identical to our conditions $(C1),(C2),(C4)$ and $(C5)-(C6)$. Since $w \in C^{1,2}_b([0,T) \times \D, \R)$, it is Lipschitz-continuous on $[0,T] \times \bar{D}_n$, $\forall n$. Combined with condition $(C7)$, it implies that the composition $(t,s,x) \rightarrow f(t,s,x,w(t,s,x))$ is uniformly H\"{o}lder-continuous on $[0,T] \times \bar{D}_n \times \R$, thus satisfying $(A3d')$. Lastly, the boundedness of $v$ from Lemma \ref{lemma_1} corresponds to $(A3e')$. Therefore, we conclude that $v$ is a bounded classical solution (i.e. $C^{1,2}_b([0,T) \times \D, \R)$) of the PDE \eqref{PDE_With_iteration_w_v} for $(t,s,x) \in [0,T) \times D$. 

Now, let's select the initial function $P^{(0)} \in C^{1,2}_b([0,T] \times \D, \R)$, e.g. $P^{(0)} = 0$. Then, the subsequent functions $P^{(n)} = \M P^{(n-1)}$, $n = 1,2,\ldots$, are also $C^{1,2}_b([0,T] \times \D, \R)$ and satisfy the linear inhomogeneous PDE \eqref{PDE_With_iteration}. By Proposition \ref{prop:contraction_mapping},  the contraction mapping $\M$ ensures the sequence $(P^{(n)})$ to converge to a unique fixed point $P \in C_b([0,T) \times \D, \R)$.

\subsection{Proof of Proposition \ref{Prop:UCVA_Positive_NPV}}
First, we denote $\tlambda := \lambda^{(1)} + \lambda^{(2)}$ and $\th(t,s) := h(t,s) + \lambda^{(0)} \, l(t,s)$. Applying the positive payoff to the definition of $P^b = P$ in \eqref{Expectation_With}, the MtM value with CR provision is given by   
\begin{align}
	P^b(t,s) &= \E_{t,s} \bigg[ e^{- (r + \lambda) \, (T-t)} \, g(S_T) +  \int_t^T  e^{-  (r + \lambda) \, (u-t)}  \th(u,S_u)\, du + \int_t^T (\tlambda - \alpha) \, e^{- (r + \lambda) \, (u-t)}  P^b(u,S_u) \, du \bigg]. \label{interim:Pb}
\end{align}

To prove  \eqref{with_option_buyer}, we substitute it into the RHS of \eqref{interim:Pb} and verify that it indeed reduces to  \eqref{with_option_buyer}.   To this end, we get 
\begin{align*}
	P^b(t,s) &= \E_{t,s} \bigg[ e^{- (r + \lambda) \, (T-t)} \, g(S_T) +  \int_t^T  e^{- (r + \lambda) \, (u-t)}  \th(u,S_u)\, du \notag\\  & \qquad \qquad + \int_t^T (\tlambda - \alpha) \, e^{-  (r + \lambda) \, (u-t)}   \left\{ e^{-  (r + \alpha + \lambda^{(0)}) (T-u)} \, g(S_T) +  \int_u^T  e^{- (r + \alpha + \lambda^{(0)}) \, (v-u)} \th(v,S_v) \, dv  \right\} \, du \, \bigg] \\
	&= \E_{t,s} \bigg[ e^{- (r + \alpha + \lambda^{(0)}) \, (T-t)} \, g(S_T)  \\
	& \qquad \qquad +  \int_t^T  e^{- (r + \lambda) \, (u-t)}  \th(u,S_u)\, du +  \int_t^T \int_t^v  (\tlambda - \alpha) \, e^{-  (\tlambda - \alpha) \, u - r \, (v-t) + (-\alpha - \lambda^{(0)}) \, v + \lambda \, t} \th(v,S_v) \, du \, dv \, \bigg]\\
	&= \E_{t,s} \bigg[ e^{- (r + \alpha + \lambda^{(0)}) \, (T-t)} \, g(S_T) \\
	& \qquad \qquad +  \int_t^T  e^{- (r + \lambda) \, (u-t)}  \th(u,S_u)\, du + \int_t^T  ( e^{- (r + \alpha + \lambda^{(0)}) \, (u-t)} - e^{- (r + \lambda) \, (u-t)} )  \th(u,S_u) \, du \bigg] \\
	&= \E_{t,s} \bigg[ e^{- (r + \alpha + \lambda^{(0)}) \, (T-t)} \, g(S_T) +  \int_t^T  e^{- (r + \alpha + \lambda^{(0)}) \, (u-t)}  \th(u,S_u)\, du \, \bigg].
\end{align*}
Since the last equality resembles \eqref{with_option_buyer}, we conclude. The same steps will yield the proof for expression \eqref{with_option_seller}.

To verify \eqref{without_option_buyer}, we use the expressions of $\Pi$ in \eqref{Expectation_Pi} and  $\hP^b = \hP$ in \eqref{Expectation_Without} to  get
\begin{align*}
 &\hP^b(t,s) = \E_{t,s} \bigg[ e^{- (r + \lambda) (T-t)} \, g(S_T)  +  \int_t^T  e^{- (r + \lambda) \, (u-t)}  \th(u,S_u) \, du   - \int_t^T \alpha  \,  e^{-  (r + \lambda) \, (u-t)} \,  \Pi(u,S_u) \, du \\
 & \qquad \qquad + \int_t^T \tlambda \, e^{-  (r + \lambda) \, (u-t)} \, \Pi(u,S_u) du \bigg]  \\
  &=\E_{t,s} \bigg[ e^{- (r + \lambda) (T-t)} \, g(S_T)  +  \int_t^T  e^{- (r + \lambda) \, (u-t)}  \th(u,S_u) \, du  - \int_t^T \alpha  \,  e^{-  (r + \lambda) \, (u-t)} \,  \Pi(u,S_u) \, du \\
  & \qquad \qquad +  \int_t^T \tlambda \, e^{-  (r+ \lambda) \, (u-t)}  \left[ \int_u^T e^{- (r + \lambda^{(0)}) (v-u)} \, \th(v,S_v) \, dv + e^{- (r + \lambda^{(0)}) (T-u)} \, g(S_T)   \right] \, du \, \bigg]  \\
  &=\E_{t,s} \bigg[ e^{- (r + \lambda^{(0)}) (T-t)} \, g(S_T) - \int_t^T \alpha  \,  e^{-  (r + \lambda) \, (u-t)} \,  \Pi(u,S_u) \, du\\
  & \qquad \qquad +  \int_t^T  e^{- (r + \lambda) \, (u-t)}  \th(u,S_u) \, du +  \int_t^T   (1 - e^{-  \tlambda \, (v-t)})   e^{- (r + \lambda^{(0)}) (v-t)} \, \th(v,S_v)
  \, dv \,  \bigg]  \\
    &=\E_{t,s} \bigg[ e^{- (r + \lambda^{(0)}) (T-t)} \, g(S_T) - \int_t^T \alpha  \,  e^{-  (r + \lambda) \, (u-t)} \,  \Pi(u,S_u) \, du + \int_t^T  e^{- (r + \lambda^{(0)}) \, (u-t)}  \th(u,S_u) \, du \bigg]  \\
 &= \Pi(t,s) - \E_{t,s}\left[ \int_t^T \alpha  \,  e^{-  (r + \lambda) \, (u-t)} \,  \Pi(u,S_u) \, du \right] .
\end{align*}
Applying the same steps to the definition of $\hP^s$ in \eqref{Expectation_Without_seller}, we obtain the equation \eqref{without_option_seller}.

\subsection{Proof of Proposition \ref{prop:bid_ask_dominance}}\label{sect-prof46}
From \eqref{without_option_buyer} and \eqref{without_option_seller} and the condition $\alpha, \beta \ge 0$, we obtain the inequalities $\hP^b(t,s) \le \Pi(t,s)$ and $\hP^s(t,s) \le \Pi(t,s)$.
The price expressions \eqref{with_option_buyer} and \eqref{with_option_seller} imply that
\begin{align}
	 P^b(t,s) &= \E_{t,s} \bigg[ e^{- \int_t^T (r + \alpha + \lambda^{(0)}) \, du} \, g(S_T) +  \int_t^T  e^{- \int_t^u (r + \alpha + \lambda^{(0)}) \, dv} \th(u,S_u) \, du \bigg] \notag\\
	  &\le \E_{t,s} \bigg[ e^{- \int_t^T (r + \lambda^{(0)}) \, du} \, g(S_T) +  \int_t^T  e^{- \int_t^u (r + \lambda^{(0)}) \, dv} \th(u,S_u) \, du \bigg] = \Pi(t,s). \label{Pbinequality}
\end{align}
Similar arguments give $P^s(t,s) \le \Pi(t,s) $. Hence, we conclude \eqref{bid_ask_relationship}.

Next, applying the definition of $\hP^b \equiv \hP$ in \eqref{Expectation_Without} along with the inequality \eqref{Pbinequality}, we get
\begin{align*}
 \hP^b(t,s) &= \E_{t,s} \bigg[ e^{- (T-t) \, (r + \lambda)} \, g(S_T) +  \int_t^T  e^{- (u-t) \, (r + \lambda) } \th(u,S_u) \, du + \int_t^T (\tlambda - \alpha) \, e^{- (u-t) \, (r + \lambda) } \Pi(u,S_u)\, du \bigg] \\
 & \ge \E_{t,s} \bigg[ e^{- (T-t) \, (r + \lambda)} \, g(S_T) +  \int_t^T  e^{- (u-t) \, (r + \lambda) } \th(u,S_u) \, du + \int_t^T (\tlambda - \alpha) \, e^{- (u-t) \, (r + \lambda) } P^b(u,S_u) \, du \bigg]\\
 &= P^b(t,s).
\end{align*}
The last equality follows from the definition of $P^b$ in \eqref{Expectation_With}. From the definition of $\hP^s$ in \eqref{Expectation_Without_seller} and the inequality $0 \le P^s \le \Pi$ in \eqref{bid_ask_relationship}, we obtain
\begin{align*}
 \hP^s(t,s) &= \E_{t,s} \bigg[ e^{- (T-t) \, (r + \lambda)} \, g(S_T) +  \int_t^T  e^{- (u-t) \, (r + \lambda) } \th(u,S_u) \, du + \int_t^T (\tlambda - \beta) \, e^{- (u-t) \, (r + \lambda) } \Pi(u,S_u) \, du \bigg] \\
 & \ge \E_{t,s} \bigg[ e^{- (T-t) \, (r + \lambda)} \, g(S_T) +  \int_t^T  e^{- (u-t) \, (r + \lambda) } \th(u,S_u) \, du + \int_t^T (\tlambda - \beta) \, e^{- (u-t) \, (r + \lambda) } P^s(u,S_u) \, du \bigg]\\
 &= P^s(t,s).
\end{align*}
The last equality holds from the definition of  $P^s$ in \eqref{Expectation_With_seller}. Hence, we conclude \eqref{bid_ask_relationship2}. 

\subsection{CRF CDS and TRS Prices with  OU and CIR Reference Default Intensities}
Here we summarize the  computation of the CRF values of CDS and TRS.  Let us assume that the risk free rate is a time-deterministic function $r(t)$, and the default intensities $\lambda^{(i)}$, $i \in \{0,1,2\}$, are of the form $\lambda^{(i)}(t,X_t)  = \psi_i(t) + w_{i} \, X_t$, where $X$ follows the OU or CIR model.  This   specification is more general than the one used in the numerical examples in Sections \ref{sect-cds}-\ref{sect-trs}, and it  yields  analytic expressions for the pre-default prices (without counterparty risk provision) of the credit default swaps and total return swaps considered therein. These  prices are used for comparison analysis (see Figures \ref{BCVA_CDS_Convergence}  and \ref{BCVA_TRS_CIR_OU}).  

 In the OU model, the pre-default zero-coupon bond price with maturity $T$ and zero recovery meets the formula \cite[Chap.~7.1.1]{schonbucher2003credit}
\begin{align}
 C_1(t,x\,;\,T) &= e^{- \int_t^T r(u) \, du} \, e^{- \int_t^T \psi_0(u) du} \, \E_{t,x}\left[ e^{- \int_t^T  w_{0} \, X_s \, ds} \right] \nonumber \\
& = e^{- \int_t^T (r(u) + \psi_0(u)) du} \, e^{A_1(T-t) - B_1(T-t) \, w_0 \, x}, \label{eq:zero_coupon_bond_OU}
\end{align}
where
\begin{align*}
 A_1(u) = \int_0^u \big( \frac{1}{2} \sigma^2 B_1(v)^2 - \kappa \, \theta \, B_1(v) \big) dv\,, \quad B_1(u) = \frac{1- e^{\kappa \, u}}{\kappa}, \quad 0 \le u \le T. 
\end{align*}
In the CIR model, the bond price is given by \cite[Chap.~7.2]{schonbucher2003credit}
\begin{align}
 C_2(t,x\,;\,T) &= e^{-\int_t^T r(u) \, du} \, e^{- \int_t^T \psi_0(u) du} \E_{t,x}\left[ e^{- \int_t^T  w_{0} \, X_u \, du} \right] \nonumber \\
& = e^{- \int_t^T (r(u) + \psi_0(u)) du} \, A_2(T-t) e^{- B_2(T-t) \, x}, \label{eq:zero_coupon_bond}
\end{align}
where 
\begin{align}
&A_2(u) = \left[   \frac{2 \Xi \, e^{\frac{u}{2} \, (\Xi + \kappa)} }{( \Xi + \kappa) \, (e^{\Xi \cdot u} - 1) + 2 \Xi} \right]^{ \frac{2 \kappa \, \theta}{\sigma^2}} \,,\quad B_2(u) = \left[ \frac{2 \, (e^{\Xi \cdot u} - 1) \, w_{0}}{( \Xi + \kappa) \, (e^{\Xi \cdot u} - 1) + 2 \Xi}\right],\label{CDS_A_B}
\end{align}
for $0 \le u \le T$ with constant $\Xi = \sqrt{\kappa^2 + 2 \, \sigma^2 \, w_{0}}$\,.

In turn, we can  summarize the prices for the CDS discussed  in Section \ref{sect-cds}. Under the OU model, the buyer pays the pre-default \textit{upfront price} for  the CDS: 
\begin{align}
 \Pi(t,x) &= \int_t^T C_1(t,x\,;\,u) \left( w_0 \, x e^{- \kappa \,(u - t)} + (\kappa \, \theta - \frac{\sigma^2}{\kappa}) (u-t) + \frac{\sigma^2}{\kappa} (1 - e^{-\kappa (u-t)} ) - p \right) \, du\,,
\end{align}
 where $C_1(t,x\,;\,u)$  is the pre-default zero coupon bond price with maturity $T$ and zero recovery given in  \eqref{eq:zero_coupon_bond_OU}.  Under the CIR model, the pre-default upfront price of a  CDS  with maturity  $T$ and  premium rate $p$ is given by 
\begin{align}
\Pi(t,x) = \int_t^T C_2(t,x;u) \left[ w_{0} (\kappa \, \theta \, B_2(u-t) + B_2^{'}(u-t) x) - p \right] du\,, \label{eq:CDS}
\end{align}
with  $C_2(t,x\,;\,u)$  in \eqref{eq:zero_coupon_bond} and $B_2(\cdot)$  in \eqref{CDS_A_B}. See Chap. 7 of \cite{schonbucher2003credit}.

As for the total return swap described in Section \ref{sect-trs}, its  CRF upfront value at time $t \le T$,  is given by
\begin{align}
 \Pi(t,x) &= \E_{t,x}\bigg[ e^{-\int_t^T (r(u) + \lambda^{(0)}(u,X_u)) \, du} ( C(T,X_{T}\,;\,T') - K ) \notag \\
 &~~ - \int_t^T \lambda^{(0)}_u \, e^{-\int_t^u  (r(v) + \lambda^{(0)}(v,X_v)) \, dv} \,  K \, du - \int_t^T e^{-\int_t^u (r(v) + \lambda^{(0)}(v,X_v)) \, dv} \, K \, (r(u) + p) \, du \bigg]\notag\\
 &= C(t,x\,;\,T') - K \left( 1 + p \int_t^T C(t,x;u) du \right) .\label{eq:TRS}
\end{align}
For the  CRF prices  in Figures    \ref{BCVA_CDS_Convergence}  and \ref{BCVA_TRS_CIR_OU}), we set $\psi_0 = 0$, and  assume constant interest rate $r$. The counterparty  default rates $\lambda^{(1)}$ and $\lambda^{(2)}$ are also constant, and  the reference  default rate is set to be $\lambda^{(0)}(t,X_t) = X_t$.  This  amounts to setting  $w_1 = w_2 = 0$, and $w_0 = 1$  in \eqref{eq:zero_coupon_bond_OU}-\eqref{eq:TRS} to obtain the CRF values.
\begin{small}
\bibliographystyle{apa}
\bibliography{reference}
\end{small}
\end{document}